\documentclass[conference,letterpaper]{IEEEtran}
\usepackage[utf8]{inputenc} 
\usepackage[T1]{fontenc}
\usepackage{url}
\usepackage{ifthen}
\usepackage{cite}
\usepackage[cmex10]{amsmath} 
\usepackage{amsthm}
\usepackage{amsmath}
\usepackage{amsfonts}
\usepackage{amssymb}
\usepackage{color}
\usepackage{graphicx} 
\usepackage{stackrel}
\usepackage{booktabs}
\usepackage{enumitem}
\usepackage{balance}
\theoremstyle{plain}
\newtheorem{theorem}{Theorem}
\newtheorem{claim}{Claim}
\newtheorem{corollary}{Corollary}
\newtheorem{lemma}{Lemma}
\newtheorem{definition}{Definition}
\newtheorem{statement}{Statement}
\usepackage[utf8]{inputenc}
\newcommand{\In}{{\rm In}}
\newcommand{\Out}{{\rm Out}}

\title{A Local Perspective on the Edge Removal Problem}
\author{\IEEEauthorblockN{Fei Wei}
\IEEEauthorblockA{University at Buffalo\\
feiwei@buffalo.edu}
\and
\IEEEauthorblockN{Michael Langberg}
\IEEEauthorblockA{University at Buffalo \\
mikel@buffalo.edu}
\and
\IEEEauthorblockN{Michelle Effros}
\IEEEauthorblockA{California Institute of Technology\\
effros@caltech.edu}
}

\begin{document}
\maketitle


\begin{abstract}
    The edge removal problem studies the loss in network coding rates that results when a network communication edge is removed from a given network. 
    It is known, for example, that in networks restricted to linear coding schemes and networks restricted to Abelian group codes, removing an edge $e^*$ with capacity $R_{e^*}$ reduces the achievable rate on each source by no more than $R_{e^*}$.
    In this work, we seek to uncover larger families of encoding functions for which the edge removal statement holds.
    We take a {\em local} perspective: instead of requiring that all network encoding functions satisfy certain restrictions (e.g., linearity), we limit only the function carried on the removed edge $e^*$. 
    Our central results give sufficient conditions on the function carried by edge $e^*$ in the code used to achieve a particular rate vector under which we can demonstrate the achievability of a related rate vector once $e^*$ is removed.
\end{abstract}


\section{Introduction}

The {\em edge removal problem} (see, for example,
\cite{ho2010equivalence, jalali2011impact, langberg2011network, wong2013capacity, lee2013outer, wong2015equivalence, wong2016tightnes, wei2017effect,gu2009achievable,kosut2016relationship}) studies the impact on communication rate of removing a network communication edge $e^*$ with capacity $R_{e^*}$ from a given network coding network.
(Detailed definitions of the statement above and those that follow appear in Section~\ref{sec:model}.)
In a number of special cases, including networks with co-located sources, networks with a super-source node, networks where the removed edge is connected to a terminal node, networks restricted to linear coding schemes, and networks restricted to Abelian group codes
\cite{ho2010equivalence, jalali2011impact, langberg2011network,wei2017effect}, it is known that removing an edge $e^*$ with capacity $R_{e^*}$ from a given network reduces the achievable rate on each source by no more than the edge capacity $R_{e^*}$.
When the loss in each source's rate for each source that results from an edge removal is bounded by the capacity of the edge removed, we say that the network satisfies an {\em edge removal statement}.

Whether the edge removal statement holds universally for any network coding instance and any family of encoding functions remains an intriguing open question with connections to a variety of different network information theoretic problems. For example, the edge removal problem lends insight on the relationship between index coding and network coding instances \cite{wong2013capacity}, the reducibility of multiple-multicast network coding to 2-unicast network coding \cite{kamath2014two,wong2015equivalence}, the entropic region characterization of network coding instances \cite{chan2014network,kim2016characterization,wong2016tightnes}, other upper bounds on network coding capacity \cite{lee2013outer}, and the notion of strong converses in network communication \cite{gu2009achievable,kosut2016relationship}. These rich connections shed light on the significance of the edge-removal problem and its solution.

In this work, we continue the line of study from \cite{lee2013outer, jalali2011impact,wei2017effect}, seeking to find larger families of encoding functions for which an edge removal statement holds. We here introduce a new {\em local} perspective, observing that when considering the impact of removing a given edge $e^*$ from a network, it sometimes suffices to consider only the function carried by edge $e^*$, rather than considering the encoding functions for all network edges. 
For example, while \cite{jalali2011impact} shows the edge removal statement is satisfied on networks restricted to linear coding schemes, we show, for a given coding scheme, that an edge removal statement is true on all edges $e^*$ that carry linear functions of the source inputs even if other edges in the network carry non-linear functions of those inputs.\footnote{We note that this extension is implicit in the proofs appearing in \cite{jalali2011impact}.}
More precisely, we study a {\em local edge removal statement}: Given any network coding instance that has a solution {achieving rate vector $\mathbf{R}$}, if the function carried on a given edge $e^*$ satisfies certain conditions (e.g., linearity), removing $e^*$ from the network reduces the achievable rate on each source by no more than the capacity of $e^*$.
Our work seeks to understand for which functions carried on edge $e^*$ the local edge removal statement holds.

The remainder of this paper is organized as follows.
In Section~\ref{sec:model}, we present our model and the definitions used throughout this work.
In Theorem~\ref{claim:erY} of Section~\ref{sec:local}, we present a sufficient condition for the function carried by edge $e^*$ that implies our local edge removal statement.
As the local edge removal statement can imply the (original) edge removal statement studied, e.g., in \cite{jalali2011impact, wei2017effect}, our sufficient condition unifies and generalizes the previous results of \cite{jalali2011impact, wei2017effect} on linear and Abelian group network codes.

In Section~\ref{sec:CWL}, we define a broad class of functions that generalize linear functions. We call these {\em component-wise-linear} (CWL) functions. We study CWL functions and show that such functions satisfy the sufficient condition of Theorem~\ref{claim:erY}.
Thus, for a given code, the local edge removal statement holds on $e^*$ if the function carried by $e^*$ is CWL. 
We then expand this result to certain functions that are {\em piece-wise} CWL. 

In Section~\ref{sec:nonlinear}, we study the possibility of proving the original edge removal statement (studied, e.g., in \cite{ho2010equivalence, jalali2011impact, langberg2011network, wong2013capacity, lee2013outer, wong2015equivalence, wong2016tightnes, wei2017effect,gu2009achievable,kosut2016relationship}) through the local lens of CWL functions. Namely, we ask the following question: {\em Given a network coding instance, a network edge $e^*$, and a network coding scheme, can the scheme be modified so that (a) the function carried on edge $e^*$ is CWL, and (b) there is no compromise on the communication rate?}
In other words, we ask if one can always modify a network coding scheme to preserve communication and guarantee that ``locally'' (on a given edge $e^*$) the function is CWL, whereas elsewhere, functions may be arbitrary.
An affirmative answer to this question would imply the original edge removal statement for any network coding instance, while a negative answer would not necessarily have any implications on edge removal. 
As any instance for which linear encoding functions are optimal satisfies the conditions of the question above, we study instances for which linear encoding functions are sub-optimal. 
While we do not resolve the question in this work, we prove that such code modifications are possible for the network coding instances and solutions presented in \cite{dougherty2005insufficiency, chan2008dualities, connelly2017class} for which linear coding is known to be sub-optimal, implying that edge removal holds for these solutions. {We note that our question is not resolved on the network instance given in \cite{blasiak2011lexicographic} for which, as above, linear network coding is suboptimal.}

Finally, we conclude in Section~\ref{sec:conc}. The proofs of several of our claims appear in the Appendix.


\section{Model and definitions} \label{sec:model}
Throughout the paper, we denote the size of a finite set $\mathcal{S}$ by $|\mathcal{S}|$. For any positive integer $k$, we denote the set $\{1,\dots,k\}$ by $[k]$. We use bold letters to denote vectors, for example $\mathbf{R}=(R_1,\dots,R_{|\mathcal{S}|})$ is a vector of dimension $|\mathcal{S}|$ and $R_i$ is the $i^{th}$ element of vector $\mathbf{R}$. 

A network instance $\mathcal{I} = (\mathcal{N},\mathcal{S},\mathcal{T},\mathcal{M})$ of the network coding problem includes a directed acyclic error-free network $\mathcal{N} = G(\mathcal{V,E})$ with nodes (also refered to as vertices) $\mathcal{V}$ and edges $\mathcal{E} \subset \mathcal{V}\times\mathcal{V}$. Each edge $e=(v_1,v_2)\in\mathcal{E}$ represents an error-free channel from node $v_1$ to node $v_2$. We use $R_e>0$ to denote the channel capacity of edge $e$. For each node $v\in\mathcal{V}$, $ \In(v) = \{(v_1,v):(v_1,v)\in\mathcal{E}\}$ and $\Out(v) = \{(v,v_1):(v,v_1)\in\mathcal{E}\}$ denote the set of incoming and outgoing edges of node $v$ respectively. The sets $\mathcal{S}\subseteq \mathcal{V}$ and $\mathcal{T}\subseteq \mathcal{V}$ denote the set of source nodes and the set of terminal nodes respectively. Without loss of generality, each source $s\in\mathcal{S}$ has no incoming edges and each terminal $t\in\mathcal{T}$ has no outgoing edges, which implies $\mathcal{S}\cap\mathcal{T}=\phi$. For convenience of notation, at times, we denote nodes in the source set $\mathcal{S}$ by integers $i \in \{1,2,\dots,|\mathcal{S}|\}$. Finally, $|\mathcal{S}|\times|\mathcal{T}|$ binary matrix $\mathcal{M} = [m_{st}]$ describes the network demands, with $m_{st} = 1$ if and only if source $s \in \mathcal{S}$ is requested by terminal $t \in \mathcal{T}$. 

Consider an instance $\mathcal{I} = (\mathcal{N},\mathcal{S},\mathcal{T},\mathcal{M})$ of the network coding problem. A network code of blocklength $n$ and rate vector $\mathbf{R} = (R_i)_{i\in \mathcal{S}}$ on $\mathcal{I}$ is defined \emph{by a set of random variables $\{X_f : f\in \mathcal{S}\cup\mathcal{E}\}$} as follows. 

Each source $i\in\mathcal{S}$ independently generates source message $X_i$ uniformly distributed on the support set $\mathcal{X}_i = [2^{nR_i}]$. Each edge $e\in\mathcal{E}$ carries edge message $X_e$ with support set $\mathcal{X}_e = [2^{nR_e}]$. For any set $A\subseteq \mathcal{E}$, we use $X_A$ to represent the vector of messages on the edges in $A$, giving $X_A=(X_e)_{e\in A}$. The \textit{local encoding function} $\phi_{le}:\prod_{f\in \In(u)}\mathcal{X}_f \mapsto \mathcal{X}_e$ for edge $e=(u,v)\in \mathcal{E}$ takes as its input the vector $X_{\In(u)}$ of messages (i.e., random variables) associated with incoming edges $\In(u)$. The edge message $X_e$ equals the evaluation of $\phi_{le}$ on its input, giving $X_e=\phi_{le}(X_{\In(u)})$. Notice that this relationship implies $H(X_e|X_{\In(u)})=0$, for all $e=(u,v)\in \mathcal{E}$. 
Edges that leave source node $s\in\mathcal{S}$ have corresponding local encoding functions that take the source information $X_s$ as input.

The \textit{global encoding function} may be defined inductively from the local encoding functions, see e.g.\cite{yeung2008information}. For any edge $e=(u,v)\in\mathcal{E}$, edge message $X_e$ is uniquely determined by source message vector $X_\mathcal{S}=(X_i)_{i\in\mathcal{S}}$ via the global encoding functions $\phi_{ge}: \prod_{i\in\mathcal{S}} \mathcal{X}_i \mapsto \mathcal{X}_e,$ giving $X_e=\phi_{le}(X_{\In(u)})=\phi_{ge}(X_\mathcal{S})$.

For any terminal node $t\in\mathcal{T}$, the received message $X_{\In(t)}$ equals the evaluation of the global encoding functions $\phi_{ge}$ on $e\in \In(t)$; thus $H(X_{\In(t)}|X_{\mathcal{S}})=0$. The decoding function 
$\phi_{t}:\prod_{e\in\In(t)}\mathcal{X}_e \mapsto \hat{\mathcal{X}}_t$ takes as input received message $X_{\In(t)}$ and emits the reconstruction $\hat{X}_t$ of the demanded source messages, giving $\hat{X}_t=\phi_t(X_{\In(t)})$.

An instance $\mathcal{I}$ is said to be \emph{$(\epsilon, \mathbf{R}, n)$-feasible} if and only if there exists a network code $\{X_f:f\in\mathcal{S}\cup\mathcal{E}\}$ (with corresponding encoding/decoding functions) of blocklength $n$, that satisfies the following properties:
\begin{enumerate}
    \item Uniform and independent sources: $H(X_\mathcal{S}) = \sum_{i\in\mathcal{S}}H(X_i)=\sum_{i\in\mathcal{S}}\log |\mathcal{X}_i|$.
    \item Source rate: $H(X_i)\geq n R_i$, $\forall i\in\mathcal{S}$.
    \item Edge capacity: $\mathcal{X}_e = [2^{n R_e}]$, $\forall e\in \mathcal{E}$.    
    \item Encoding: $H(X_{e}|X_{\In(u)}) = 0$, $\forall e=(u,v)\in \mathcal{E}$.
    \item Decoding: for each terminal $t\in\mathcal{T}$, $\Pr(\hat{X}_t = X_t) > 1- \epsilon$, here, the probability is taken with respect to the distribution on the source random variables $X_\mathcal{S}$.
\end{enumerate}


\begin{statement}[Local edge removal statement] \label{statement:edgeremoval}
    Let $\mathcal{I}=(\mathcal{N},\mathcal{S},\mathcal{T},\mathcal{M})$ be a network coding instance. Suppose that $\mathcal{I}$ is $(\epsilon,\mathbf{R},n)$-feasible. Consider a modified instance $\mathcal{I'} = (\mathcal{N'},\mathcal{S},\mathcal{T},\mathcal{M})$ obtained by removing an edge $e^*$ with capacity $R_{e^*}$ from $\mathcal{N}$. If $\mathcal{I}'$ is $(\epsilon,\mathbf{R}-R_{e^*}\cdot\mathbf{1},n)$-feasible, we say that the local edge removal statement holds on $\mathcal{I}$ for edge $e^*$ and parameters $(\epsilon,\mathbf{R},n)$. If the local edge removal statement holds in this setting for all edges in $\mathcal{N}$ and all settings of $(\epsilon,\mathbf{R},n)$ for which $\mathcal{I}$ is $(\epsilon,\mathbf{R},n)$-feasible, we say that the edge removal statement holds on $\mathcal{I}$. Here $\mathbf{R}-R_{e^*}\cdot\mathbf{1}=((R_1-R_{e^*})^+,\dots,(R_{|\mathcal{S}|}-R_{e^*})^+)$ where $(R-R_{e^*})^+=\max\{0,R-R_{e^*}\}$.
\end{statement}


\section{A sufficient condition for local edge removal} \label{sec:local}

In what follows, we present our sufficient condition for local edge removal. Roughly speaking, we show that the local edge removal statement holds on $\mathcal{I}$ for a given edge $e^*$ in $\mathcal{I}$ and parameters $(\epsilon,\mathbf{R},n)$ if there exists an $(\epsilon,\mathbf{R},n)$-feasible code $\{X_f: f\in\mathcal{S}\cup \mathcal{E}\}$ for $\mathcal{I}$ and an auxiliary random variable $Y$ that is a deterministic function $f_Y$ of $X_{\mathcal{S}}$ such that:
\begin{enumerate}[label=(\Alph*)]
    \item $X_{e^*}$ is a deterministic function of $Y$.
    \item Conditioned on $Y$ the sources are independent.
    \item There exists a value $y'$ of $Y$ such that (a) the (source) pre-image of $y'$ under function $f_Y$ is ``relatively large'', and (b) the pre-image captures the error statistics of the network code.
\end{enumerate}
Note that as $Y$ is a deterministic function of $X_\mathcal{S}$, the random variable $Y$ induces a partition of the set $\mathcal{X}_\mathcal{S}$. Requirements $(A)-(C)$ are formalized in the theorem below using the following notation. Given an $(\epsilon,\mathbf{R},n)$ network code, we say $x_\mathcal{S}$ is ``good'' if $x_\mathcal{S}$ is decoded correctly by all terminals; otherwise $x_\mathcal{S}$ is ``bad.'' We define the set of all good elements in $\mathcal{X}_\mathcal{S}$ as $\mathcal{X}^G_\mathcal{S}=\{x_\mathcal{S}: x_\mathcal{S}\ \mbox{is good}\}$ and the set of bad elements as $\mathcal{X}^B_\mathcal{S}=\mathcal{X}_\mathcal{S}\setminus \mathcal{X}^G_\mathcal{S}$. With these definitions, $\epsilon = \frac{|\mathcal{X}^B_\mathcal{S}|}{|\mathcal{X}_\mathcal{S}|}=\frac{|\mathcal{X}^B_\mathcal{S}|}{|\mathcal{X}^G_\mathcal{S}|+|\mathcal{X}^B_\mathcal{S}|}$ by the assumed uniform distribution on $\mathcal{X}_\mathcal{S}$.
We also express the partition induced by $y$ as $A(y)=\{x_\mathcal{S}\in \mathcal{X}_\mathcal{S}:f_Y(x_\mathcal{S})=y\}$, $A^G(y)=\{x_\mathcal{S}\in\mathcal{X}^G_\mathcal{S}:f_Y(x_\mathcal{S})=y \}$, $A^B(y)=A(y)\setminus A^G(y)$ and $A_i(y)=\{x_i:\exists x_\mathcal{S}\in A(y), x_i=x_\mathcal{S}(i) \}$ for any $i\in\mathcal{S}$.


\begin{theorem} \label{claim:erY} 
    Let $\mathcal{I}=(\mathcal{N},\mathcal{S},\mathcal{T},\mathcal{M})$ be a network coding instance that is $(\epsilon,\mathbf{R},n)$-feasible with a corresponding network code $\{X_f: f\in \mathcal{S}\cup \mathcal{E}\}$. The local edge removal statement holds on $\mathcal{I}$ for edge $e^*$ in $\mathcal{I}$ and parameters $(\epsilon,\mathbf{R},n)$ if there exists a random variable $Y$ (with support $\mathcal{Y}$) which is a deterministic function $f_Y$ of $X_{\mathcal{S}}$ such that
    \begin{enumerate}[label=(\Alph*)]
        \item $H(X_{e^*}| Y) = 0$. (Therefore, there exists a function $g_Y$ s.t. $X_{e^*}=g_Y(Y)$.)
        \item $H(X_\mathcal{S}|Y)=\sum_{i\in\mathcal{S}}H(X_i|Y)$.
        \item There exists $y'\in\mathcal{Y}$ s.t. $|A_i(y')|\geq |\mathcal{X}_i|/|\mathcal{X}_{e^*}|$ for every $i\in\mathcal{S}$ and $|A^B(y')|\leq \epsilon|A(y')|$.
    \end{enumerate}
\end{theorem}

\begin{proof}
By assumption (A), $X_{e^*}=g_Y(Y)$ for some deterministic function $g_Y$. Fix any $y'\in\mathcal{Y}$ that satisfies condition (C) of the theorem, let $x_{e^*}=g_Y(y')$. Fixing $Y=y'$, we design a new network code for the network instance $\mathcal{I}'$ obtained by removing edge $e^*$ from instance $\mathcal{I}$. The new network code is the restriction of the original code to $A(y')$.

We first define the source random variable $X'_\mathcal{S}$ with alphabet $\mathcal{X}_\mathcal{S}'=A(y')$ and probability mass function (pmf) $\Pr(X_\mathcal{S}'=x_\mathcal{S}') =\Pr(X_\mathcal{S}=x_\mathcal{S}'|Y=y')=\frac{1}{|A(y')|}$.
Under this pmf, $H(X_\mathcal{S}')=H(X_\mathcal{S}|Y=y')=\log |A(y')|$. 
We define the random variable $X_i'$ with alphabet $\mathcal{X}_i'=A_i(y')$, and pmf $\Pr(X_i'=x_i') =\sum_{x_{\mathcal{S}}' \in \mathcal{X}_{\mathcal{S}}': x_{\mathcal{S}}'(i)=x_i'} \Pr(X_{\mathcal{S}}'=x_{\mathcal{S}}')=\Pr(X_i=x_i'| Y=y')$, similarly, $H(X_i')=H(X_i|Y=y')$.

For any $e\in\mathcal{E}$, let $X'_e$ be a random variable with alphabet $\mathcal{X}'_e=\{x_e\in\mathcal{X}_e: \exists x_\mathcal{S} \in \mathcal{X}'_\mathcal{S}, x_{e}=\phi_{g e}(x_\mathcal{S})\}$ and pmf $\Pr(X_e'=x_e')=\Pr(X_e=x_e'|Y=y')$. Note that 
\begin{equation} \label{equa:YsizeXe'}
    |\mathcal{X}_e'|\leq|\mathcal{X}_e|.
\end{equation}
We now show that $\{X'_f: f \in \mathcal{S}\cup\mathcal{E}\}$ is a network code over instance $\mathcal{I}'$.

\emph{1) Uniform and independent sources:} 
We show independence first. By definition, we have
\begin{equation}\label{equa:YindS1}
     H(X_\mathcal{S}|Y)=\sum_{y\in\mathcal{Y}} \Pr(Y=y) H(X_\mathcal{S}|Y=y)
\end{equation}
and
\begin{equation} \label{equa:YindS2}
\begin{aligned}
    \sum_{i\in\mathcal{S}} H(X_i|Y) =& \sum_{i\in\mathcal{S}}\sum_{y\in\mathcal{Y}} \Pr(Y=y) H(X_i|Y=y)\\=& \sum_{y\in\mathcal{Y}} \Pr(Y=y) \sum_{i\in\mathcal{S}} H(X_i|Y=y)
\end{aligned}  
\end{equation}

Combining (\ref{equa:YindS1}), (\ref{equa:YindS2}) and condition (B) of the theorem gives 
\begin{equation} \label{equa:YindS3}
{\begin{aligned}
    \sum_{y\in\mathcal{Y}} \Pr(Y=y) H(X_\mathcal{S}|Y=y)= \sum_{y\in\mathcal{Y}} \Pr(Y=y) \sum_{i\in\mathcal{S}} H(X_i|Y=y).
\end{aligned}}
\end{equation}
For any $y \in \mathcal{Y}$, $H(X_\mathcal{S}|Y=y) \leq \sum_{i\in\mathcal{S}} H(X_i|Y=y)$; therefore, (\ref{equa:YindS3}) implies 
\begin{equation}\label{equa:condition}
    H(X_\mathcal{S}|Y=y) = \sum_{i\in\mathcal{S}} H(X_i|Y=y)
\end{equation}
for each $y\in\mathcal{Y}$. Setting $Y=y'$, we conclude that $\{X_i':i\in\mathcal{S}\}$ are independent.

Now we show that $\{X_i':i\in\mathcal{S}\}$ are uniform. By our definitions, for any $x_{\mathcal{S}}=(x_i)_{i\in\mathcal{S}}\in A(y')$, we have $x_i\in A_i(y')$ for all $i\in\mathcal{S}$, but not necessarily vice versa. Thus, $|A(y')|\leq \prod_{i\in\mathcal{S}}|A_i(y')|$.
Suppose that there exists a tuple $x_{\mathcal{S}}=(x_i)_{i\in\mathcal{S}}$ such that $x_i\in A_i(y')$ for all $i\in\mathcal{S}$ but $x_{\mathcal{S}}\not\in A(y')$, namely, $|A(y')|< \prod_{i\in\mathcal{S}}|A_i(y')|$. That is, $\Pr(X_\mathcal{S}=x_\mathcal{S}|Y=y')=0$ and $\Pr(X_i=x_i|Y=y')\neq 0$ for all $i\in\mathcal{S}$. Accordingly $\Pr(X_\mathcal{S}=x_\mathcal{S}|Y=y')\neq \prod_{i\in\mathcal{S}}\Pr(X_i=x_i|Y=y')$ which contradicts the independence of $\{X_i':i\in\mathcal{S}\}$. Thus $|A(y')|= \prod_{i\in\mathcal{S}}|A_i(y')|$.
If there exists $i\in\mathcal{S}$ such that $X_i'$ is not uniform, then $H(X_{\mathcal{S}}')=\log |A(y')|=\sum_{i\in\mathcal{S}}\log |A_i(y')|>\sum_{i\in\mathcal{S}} H(X_i')$ which is a contradiction to the independence of $\{X_i':i\in\mathcal{S}\}$. 

\emph{2) No information on edge $e^*$:}
By condition (A), $X_{e^*}=g_Y(Y)$. Thus, the value on $e^*$ is fixed over $X'_{\mathcal{S}}$, which means we can use this code when the edge $e^*$ is removed from the network.

\emph{3) Encoding:}
By the definition of a network code, $H(X_e|X_{\In(u)})=0$ for every $e=(u,v)\in\mathcal{E}$. For the same edge $e=(u,v)$ in instance $\mathcal{I}'$, $H(X_e',X_{\In(u)}') = H(X_e,X_{\In(u)}|Y=y') =H(X_{\In(u)}|Y=y') = H(X_{\In(u)}')$. Thus $H(X_e'|X_{\In(u)}')=0$ and $\{X_f':f\in\mathcal{S}\cup\mathcal{E}\}$ is a network code. 

\emph{4) Decoding:} 
By our definitions, the decoding error is $|A^B(y')|/|A(y')|$ which by condition (C) of the theorem is at most $\epsilon$. 

\emph{5) Impact on source rate:}
By our definitions and by source uniformity, we have $n R_i'=\log |\mathcal{X}_i'|=\log |A_i(y')|$, $n R_i=\log |\mathcal{X}_i|$, and $n R_{e^*}= \log |\mathcal{X}_{e^*}|$. 
By condition (C) of the theorem, for each $i\in\mathcal{S}$ we have $R_i'\geq R_i-R_{e^*}$.

\emph{6) Edge capacity limit:}
By the definition of the new edge messages and (\ref{equa:YsizeXe'}), we have $2^{n R_e'}=|\mathcal{X}'_e| \leq |\mathcal{X}_e|= 2^{n R_e}$, accordingly $ R_e'\leq R_e$ for any $e\in\mathcal{E}$ and the edge capacity limit is satisfied.
\end{proof}


Corollary~\ref{corollary:1} follows directly from the proof of Theorem~\ref{claim:erY}.

\begin{corollary}\label{corollary:1}
    Let $\mathcal{I}$ be a network coding instance that is $(\epsilon,\mathbf{R},n)$-feasible with a corresponding network code. The local edge removal statement holds on $\mathcal{I}$ for edge $e^*$ in $\mathcal{I}$ and parameters $(\epsilon,\mathbf{R},n)$ if there exists a subset $\mathcal{X}_{\mathcal{S}}'\subseteq \mathcal{X}_{\mathcal{S}}$ such that $\mathcal{X}_{\mathcal{S}}'=\prod_{i\in\mathcal{S}} \mathcal{X}_i'$ for $\mathcal{X}_i'\subseteq \mathcal{X}_i$, with 1) $X_{e^*}$ is constant conditioned on $X_\mathcal{S}'$, 2) $|\mathcal{X}_i'|\geq \frac{|\mathcal{X}_i|}{|\mathcal{X}_{e^*}|}$ and 3) $|\mathcal{X}_\mathcal{S}'\cap \mathcal{X}_\mathcal{S}^{G}|\geq (1-\epsilon)|\mathcal{X}_\mathcal{S}'|$.
\end{corollary}

Corollary~\ref{claim:erGroup} addresses Abelian group network codes which reproves Theorem 3 of \cite{wei2017effect} using the methodology of Theorem~\ref{claim:erY}. The proof is given in Appendix. 
\begin{corollary} \label{claim:erGroup}
    Let $\mathcal{I}$ be a network coding instance which is $(\epsilon,\mathbf{R},n)$-feasible for $\epsilon<\frac{1}{2}$ with a corresponding Abelian group network code $\{X_f:f\in\mathcal{S}\cup\mathcal{E}\}$. The edge removal statement holds on instance $\mathcal{I}$.
\end{corollary}
Corollary~\ref{claim:erXe} claims that the local edge removal statement holds on an edge $e^*$ if the edge random variable $X_{e^*}$ satisfies certain conditions.
We may notice that here we only consider $\epsilon=0$.
The detailed discussion and proof are given in Appendix.
\begin{corollary} \label{claim:erXe}
    Let $\mathcal{I}$ be a network coding instance which is $(0,\mathbf{R},n)$-feasible with a corresponding network code $\{X_f:f\in\mathcal{S}\cup\mathcal{E}\}$. The local edge removal statement holds on instance $\mathcal{I}$, edge $e^*$, and parameters $(0,\mathbf{R},n)$ if $H(X_\mathcal{S}|X_{e^*}) =\sum_{i\in\mathcal{S}}H(X_i|X_{e^*})$.
\end{corollary}


\section{CWL functions}\label{sec:CWL}

In this work we consider the following special family of encoding functions.

\begin{definition}[Coordinate-wise linear]\label{def:linear}
    For any $e\in\mathcal{E}$, a global encoding function $\phi_{ge}:(\mathcal{X}_1,\ldots,\mathcal{X}_{|\mathcal{S}|})\mapsto \mathcal{X}_e$ is called \emph{coordinate-wise linear} (CWL) if and only if there exist finite groups $\{G_f:f\in\mathcal{S}\cup\{e\}\}$, with group operation ``$\stackrel{f}{\circ}$'' defined on $G_f$, where $G_f=\mathcal{X}_f$ for $f\in\mathcal{S}$ and $G_e=support( \phi_{ge}) \subseteq \mathcal{X}_e$ is the support of $\phi_{ge}$, such that $\phi_{ge}$ is a homomorphism from $G_\mathcal{S} = \prod_{i\in\mathcal{S}} G_i$ to $G_e$. Namely, for every $(x_1,\ldots,x_{|\mathcal{S}|})$ and $(x_1',\ldots,x_{|\mathcal{S}|}')$
    we have $\phi_{ge}(x_1\stackrel{1}{\circ}x_1',\ldots,x_{|\mathcal{S}|}\stackrel{|\mathcal{S}|}{\circ} x_{|\mathcal{S}|}') =\phi_{ge}(x_1,\dots,x_{|\mathcal{S}|})\stackrel{e}{\circ} \phi_{ge}(x_1',\dots,x_{|\mathcal{S}|}')$.
    
\end{definition}


In the context of CWL functions, for any $f\in\mathcal{S}\cup\{e\}$, we denote the identity elements of $G_f$ as $\mathbf{i}_f$. For any $\alpha\subseteq \mathcal{S}$, we define $x_\alpha=(x_i)_{i\in\alpha}$ and $G_\alpha=\prod_{i\in\alpha}G_i$. We denote the source message vector $(\mathbf{i}_1,\dots,x_i,\dots,\mathbf{i}_{|\mathcal{S}|})$ (with $x_i$ in the $i$th coordinate and identity elements on all other coordinates) as $(x_i,\mathbf{i}_{\mathcal{S}\setminus i})$. In addition, we define $x_\mathcal{S}\circ x_\mathcal{S}'= (x_1\stackrel{1}{\circ} x_1'),\ldots,(x_{|\mathcal{S}|}\stackrel{|\mathcal{S}|}{\circ} x_{|\mathcal{S}|}')$ and sometimes omit the label on the operation ``$\circ$''. We define $\ker \phi_{ge}=\{x_\mathcal{S}\in G_\mathcal{S}: \phi_{ge}(x_\mathcal{S})=\mathbf{i}_e\}$ and $A(x_e)=\{x_\mathcal{S}\in G_\mathcal{S}: \phi_{ge}(x_\mathcal{S})=x_e\}$. 

\emph{Remark:} 
It is important to understand the relationship between CWL encoding functions and group characterizable encoding functions (see, for example, \cite{chan2008dualities}).
In group characterizable encoding functions as in CWL functions, the random variables involved are characterized by a collection of groups. However, the characterization differs in that, in group characterizable functions the support of the variables involved are associated with certain co-sets of the corresponding groups while in CWL functions they are associated with the group elements themselves.
Nevertheless, CWL functions are a special case of group characterizable functions (see the Appendix for a detailed proof).
We note that while \cite{wei2017effect} proves the edge removal statement for Abelian group codes, the CWL functions assumed here do not necessarily correspond to Abelian groups.

\subsection{CWL functions satisfy Theorem~\ref{claim:erY}}

We now show that CWL functions satisfy the local sufficient condition specified by Theorem~\ref{claim:erY}. Our proof conceptually follows ideas appearing in \cite{jalali2011impact}, which addresses edge removal in the context of linear codes.


\begin{theorem} \label{claim:ER_linear}
    Let $\mathcal{I}$ be $(\epsilon,\mathbf{R},n)$-feasible with a corresponding network code $\{X_f:f\in\mathcal{S}\cup\mathcal{E}\}$. The local edge removal statement holds on $\mathcal{I}$ for edge $e^*$ and parameters $(\epsilon,\mathbf{R},n)$ if the global encoding function {$\phi_{ge^*}$} on edge $e^*\in\mathcal{E}$ is CWL.
\end{theorem} 

\begin{proof}[Proof of Theorem~\ref{claim:ER_linear}]
Let $\{G_f:f\in\mathcal{S}\cup e^*\}$ be the groups corresponding to $\phi_{ge^*}$. 
Using $\phi_{ge^*}$, we define a random variable $Y$ and show that $Y$ satisfies the conditions of Theorem~\ref{claim:erY}. For each $i$, we start by partitioning the set $G_i$ into subsets. 
The number of subsets $n_i$ in our partition of $G_i$ determines $\mathcal{Y}_i=[n_i]$. 
The support set $\mathcal{Y}$ of $Y$ is set to $\prod_{i\in\mathcal{S}} \mathcal{Y}_i$. 
For each $i$, consider the equivalence relation $x_i \stackrel{i}{\sim} x'_i$ in which two elements of $G_i$ are equivalent if and only if $\phi_{ge^*}(x_i,\mathbf{i}_{\mathcal{S}\setminus i}) = \phi_{ge^*}(x'_i,\mathbf{i}_{\mathcal{S}\setminus i})$. 
Each $G_i$ is partitioned according the equivalence relation $\stackrel{i}{\sim}$. 
Denote the partition of $G_i$ by sets $A_i(1), \dots, A_i(n_i)$. 
Now, $Y=Y_{\mathcal{S}}$ is defined using a deterministic function $f_Y$ of $X_{\mathcal{S}}$ that maps each $(x_1,\dots,x_{|\mathcal{S}|})$ to the corresponding $y_\mathcal{S}=(y_1,\dots,y_{|\mathcal{S}|})$ for which $x_i \in A_i(y_i)$. 
The pre-image $A(y_\mathcal{S})$ of $y_\mathcal{S}$ under mapping $f_Y$ is the product set $\prod_{i \in \mathcal{S}}A_i(y_i)$ of $\mathcal{X}_{\mathcal{S}}$.

Notice that for any $x_{\mathcal{S}}=(x_1,\dots,x_{|\mathcal{S}|})$, $x'_{\mathcal{S}}=(x_1',\dots,x'_{|\mathcal{S}|})$ in the pre-image $A(y_{\mathcal{S}})$ of $y_{\mathcal{S}}=(y_1,\dots,y_{|\mathcal{S}|})$ it holds that $\phi_{ge^*}(x_{\mathcal{S}}) = \phi_{ge^*}(x'_{\mathcal{S}})$. This follows from the fact that $\phi_{ge^*}$ is CWL. Namely,
$ \phi_{ge^*}(x_1,\dots,x_{|\mathcal{S}|})=\phi_{ge^*}(x_1,\mathbf{i}_{\mathcal{S}-1})\circ \dots \circ \phi_{ge^*}(x_{|\mathcal{S}|},\mathbf{i}_{\mathcal{S}-|\mathcal{S}|})= \phi_{ge^*}(x'_1,\mathbf{i}_{\mathcal{S}-1})\circ \dots \circ \phi_{ge^*}(x'_{|\mathcal{S}|},\mathbf{i}_{\mathcal{S}-|\mathcal{S}|})=\phi_{ge^*}(x'_1,\dots,x'_{|\mathcal{S}|})$.
Thus, for any $y_\mathcal{S}=(y_1,\dots,y_{|\mathcal{S}|})$, the value of $X_{e^*}$ conditioned on $y_\mathcal{S}$ is fixed. This in turn implies condition (A) in Theorem~\ref{claim:erY}.

By our definition of $Y$, for every $y_\mathcal{S}$, the pre-image $A(y_\mathcal{S})$ of $f_Y$ is the product set $\prod_{i \in \mathcal{S}}A_i(y_i)$ of $\mathcal{X}_{\mathcal{S}}$. 
This implies that $H(X_\mathcal{S}|Y = y_\mathcal{S}) = \sum_{i\in\mathcal{S}}H(X_i|Y=y_\mathcal{S})$, which in turn implies condition (B) in Theorem~\ref{claim:erY}. 

Before studying condition (C) in Theorem~\ref{claim:erY} we present the following technical lemma.


\begin{lemma} \label{lemma:sameSize}
For any $i\in\mathcal{S}$ and for any $y_i, \bar{y} _i\in \mathcal{Y}_i$ it holds that $|A_i(y_i)|=|A_i(\bar{y}_i)|$.
\end{lemma}

\begin{proof}
Let $y_i\in\mathcal{Y}_i$ and fix $\bar{y}_i\in\mathcal{Y}_i$ to be such that $\phi_{ge^*}(\delta_i,\mathbf{i}_{\mathcal{S}\setminus i})=\mathbf{i}_{e^*}$ for all $\delta_i\in A_i(\bar{y}_i)$. As $y_i$ is any element of $\mathcal{Y}_i$, it suffices to show that $|A_i(y_i)|=|A_i(\bar{y}_i)|$.
Fix any $x_i\in A(y_i)$ and let $x_{e^*}=\phi_{ge^*}(x_i,\mathbf{i}_{\mathcal{S}\setminus i})$. 
As function $\phi_{ge^*}$ is CWL, for any $\delta_i\in A_i(\bar{y}_i)$ we have $\phi_{ge^*}(x_i\circ \delta_i,\mathbf{i}_{\mathcal{S}\setminus i})=\phi_{ge^*}(x_i,\mathbf{i}_{\mathcal{S}\setminus i})\circ \phi_{ge^*}(\delta_i,\mathbf{i}_{\mathcal{S}\setminus i})=x_{e^*}\circ \mathbf{i}_{e^*}=x_{e^*}$. 
Thus, $x_i\circ \delta_i\in A_i(y_i)$, implying that $|A_i(y_i)| \geq |A_i(\bar{y}_i)|$.

Fix any $x_i\in A(y_i)$.
By the discussion above, for any $\delta_i\in A_i(\bar{y}_i)$ we have $x_i\circ \delta_i\in A_i(y_i)$. 
If $|A_i(y_i)|>|A_i(\bar{y}_i)|$, there exists $x_i'\in A_i(y_i)$ such that $x_i'\neq x_i\cdot \delta_i$ for any $\delta_i\in A_i(\bar{y}_i)$.
In other words, there exists $x_i'=x_i\circ\delta_i'\in A_i(y_i)$ and $\delta_i' \in G_i\setminus A_i(\bar{y}_i)$.

Then we have $x_{e^*}=\phi_{ge^*}(x_i\circ \delta_i',\mathbf{i}_{\mathcal{S}\setminus i})=\phi_{ge^*}(x_i,\mathbf{i}_{\mathcal{S}\setminus i})\circ \phi_{ge^*}(\delta_i',\mathbf{i}_{\mathcal{S}\setminus i})=x_{e^*}\circ \phi_{ge^*}(\delta_i',\mathbf{i}_{\mathcal{S}\setminus i})\neq x_{e^*}$ which gives a contradiction.
Accordingly, for any $y_i\in \mathcal{Y}_i$, $|A_i(y_i)|=|A_i(\bar{y}_i)|$.
\end{proof}


Lemma~\ref{lemma:sameSize} implies that for all $y_\mathcal{S}\in \mathcal{Y}$ the size of the set $A(y_\mathcal{S})$ is identical. 
By the assumption of the theorem, there exist at least $(1-\epsilon)|G_{\mathcal{S}}|$ ``good'' elements in $G_\mathcal{S}$. 
Thus, by an averaging argument, there exists some $y'=y'_\mathcal{S}\in \mathcal{Y}$ such that
\begin{equation}\label{equa:smallerror}
    \frac{|A^G(y')|}{|A(y')|}\geq (1-\epsilon).
\end{equation}
Here $A^G(y')=\{x_\mathcal{S}\in A(y'):\mbox{$x_\mathcal{S}$ is good}\}$.
Equation~(\ref{equa:smallerror}) satisfies the second part of condition (C) in Theorem~\ref{claim:erY}.
For the first part of condition (C), our definitions imply $|\mathcal{Y}_i| \leq |G_{e^*}|$. In addition, $\cup_{y_i \in \mathcal{Y}_i} A_i(y_i) = G_i$. Thus, by Lemma~\ref{lemma:sameSize}, $|A_i(y')| \geq |G_i|/|G_{e^*}|$.
Here, for $y'=(y'_1,\dots y'_{\mathcal{S}})$, $A_i(y')=A_i(y_i')$.
This suffices to prove our assertion.
\end{proof}


\subsection{Piece-wise CWL}

\begin{definition}[Piece-wise CWL] \label{def:piece-linear}
    For any $e\in\mathcal{E}$, a global encoding function $\phi_{ge}$ is called $K$-piece-wise CWL if and only if there exists a set of CWL functions $\{\phi_{ge}^{(k)}:k\in [K]\}$, all over groups $\{G_i: i\in\mathcal{S}\}$ and $G_e$, and a partition $\{S^{(k)}:k\in [K]\}$ of $\prod_{i\in\mathcal{S}}G_i$ such that $\phi_{ge} (x_1,\ldots,x_{|\mathcal{S}|}) = \phi_{ge}^{(k)}(x_1,\ldots,x_{|\mathcal{S}|})$ if and only if $(x_1,\ldots,x_{|\mathcal{S}|})\in S^{(k)}$.
\end{definition}

We now study piece-wise CWL function in the context of edge-removal. We show below that a certain structured class of piece-wise CWL functions implies a slightly relaxed local edge-removal statement. The theorem is limited to the $0$-error setting. We discuss this limitation and give the proof in the Appendix.

\begin{theorem}\label{claim:ER_piecelinear}
    Let $\mathcal{I}=(\mathcal{N},\mathcal{S},\mathcal{T},\mathcal{M})$ be a network coding instance which is $(0,\mathbf{R},n)$-feasible with a corresponding network code that has a global encoding function $\phi_{ge^*}$ on a given edge $e^*$ which is $K$-piece-wise CWL. Let $\mathcal{I}'=(\mathcal{N}',\mathcal{S},\mathcal{T},\mathcal{M})$ be a new instance obtained by removing the edge $e^*$ from the network $\mathcal{N}$. If for every $k\in [K]$, the sub-domain $S^{(k)}$ of $\phi_{ge^*}^{(k)}$ equals a product space $\prod_{i\in\mathcal{S}}S_i^{(k)}$, then $\mathcal{I}'$ is $(0,\mathbf{R}',n)$-feasible, where $\mathbf{R}'=(R_i')_{i\in\mathcal{S}}$ and $R_i'\geq R_i-R_{e^*}-(\log K)/n$.
\end{theorem}


\section{On proving the edge-removal statement through the local lens of CWL functions}
\label{sec:nonlinear}

In this section we explore the possibility of proving the edge-removal statement through the local lens of CWL functions. Namely, we observe, via Theorem~\ref{claim:ER_linear}, that the following CWL statement, if true, implies the edge-removal statement.

\begin{statement}[CWL statement]\label{statement:task}
    Let $\mathcal{I}$ be $(\epsilon,\mathbf{R},n)$-feasible. Let $e^*$ be an edge in $\mathcal{I}$. Then there exists an $(\epsilon,\mathbf{R},n)$ coding scheme for $\mathcal{I}$ whose global encoding function on $e^*$ is CWL.
\end{statement}

For any instance $\mathcal{I}$ for which linear encoding functions are optimal, the CWL statement (Statement~\ref{statement:task}) above clearly holds. We thus turn to study the validity of the CWL statement on instances for which linear encoding functions are sub-optimal. For all previous works that we are aware of that present network coding instances that are $(\epsilon,\mathbf{R},n)$-feasible with non linear codes but not $(\epsilon,\mathbf{R},n)$-feasible with linear codes, we ask whether the CWL statement is true or not. Specifically, we ask if for any given edge $e^*$ the non-linear $(\epsilon,\mathbf{R},n)$-feasible network code presented in these works can be modified to give an $(\epsilon,\mathbf{R},n)$-feasible scheme in which the encoding on edge $e^*$ is CWL, implying that the CWL statement holds for the instance and rate vector under study. In the case studies from \cite{dougherty2005insufficiency,chan2008dualities,connelly2017class} we are able to modify the non-linear coding scheme appropriately, thus supporting the CWL statement. However we are not able to prove (or disprove) the same for the instance given in \cite{blasiak2011lexicographic}.
Details appear in the Appendix.
Whether the CWL statement holds for all network coding instances (implying that the edge-removal statement holds) is left open in this work.


\section{Conclusion}
\label{sec:conc}
In this work, we study the edge removal problem by taking a local perspective in which we focus solely on the function carried on the removed edge. Through our local perspective, we generalize (and unify) previous results by giving a sufficient condition that captures a broad range of settings including that of CWL and certain piece-wise CWL functions. We study the possibility in proving the edge-removal statement by modifying general coding solutions to have CWL encoding functions on a given edge $e^*$. Whether any network code can be modified as above remains open in this work, as does the correctness of the edge-removal statement.


\section*{Acknowledgment}
Work supported in part by NSF grants CCF-1526771, CCF-1527524 and CCF-1817241.

\bibliographystyle{IEEEtran}
\bibliography{bibliography}



\appendix


\section*{I. Proof of Corollary 2}

We apply the model and definitions from \cite{wei2017effect}. 

\begin{proof}[Proof of Corollary 2]
We show that for any edge $e^*\in\mathcal{E}$ there exists a random variable $Y$ which satisfies all three conditions in Theorem~\ref{claim:erY}, which implies that the edge removal statement holds on $\mathcal{I}$.
By Theorem~\ref{theorem:group} proven later in the Appendix, for $\epsilon<\frac{1}{2}$, any network instance $\mathcal{I}$ which is $(\epsilon,\mathbf{R},n)$-feasible by a group network code is also $(0,\mathbf{R},n)$-feasible.
Thus, without loss of generality, in the following discussion let $\epsilon=0$.
For any edge $e^*\in\mathcal{E}$, the global encoding function $\phi_{ge^*}$ is Abelian group characterizable such that the edge message $X_{e^*}$ and source messages $\{X_i:i\in\mathcal{S}\}$ can be characterized by a finite Abelian group $G$ with subgroups $G_{e^*}$ and $\{G_i:i\in\mathcal{S}\}$, respectively. 
Without loss of generality (see \cite{wei2017effect}) we assume $G_{\mathcal{S}}=\cap_{i\in\mathcal{S}} G_i=\{\mathbf{i}\}$, here $\mathbf{i}$ is the identity element of $G$.
Using $\phi_{ge^*}$ we define a random variable $Y$ and show that $Y$ satisfies the conditions in Theorem~\ref{claim:erY}.

Let $H_i=\cap_{j\in\mathcal{S}\setminus i} G_j$, by the proof of Theorem 3 in \cite{wei2017effect}, $G'=\prod_{i\in\mathcal{S}} (G_{e^*}\cap H_i)\subseteq G_{e^*}$ is an Abelian group.
We define $Y$ as the random variable characterized by $G'$. $X_{e^*}$ is a deterministic function of $Y$ (as $G'\subseteq G_{e^*}$). 
Thus condition (A) of Theorem~\ref{claim:erY} holds.

By the definition of group characterizable random variables (see \cite{wei2017effect}),
\begin{equation}\label{equa:apdxC2_1}
\begin{aligned}
    H(X_\mathcal{S}|Y)=&H(X_\mathcal{S},Y)-H(Y)\\=&\log \frac{|G|}{|G'\cap G_{\mathcal{S}}|}-\log \frac{|G|}{|G'|}\\=&\log |G'|.  
\end{aligned}
\end{equation}
Similarly for any $i\in\mathcal{S}$,
\begin{equation}\label{equa:apdxC2_2}
\begin{aligned}
    H(X_i|Y)=&H(X_i,Y)-H(Y)\\=&\log \frac{|G|}{|G'\cap G_i|}-\log \frac{|G|}{|G'|}\\=&\log \frac{|G'|}{|G'\cap G_i|}.
\end{aligned}
\end{equation}
By (27), (31) in \cite{wei2017effect} and (\ref{equa:apdxC2_1}), (\ref{equa:apdxC2_2}), we have $ H(X_\mathcal{S}|Y)=\sum_{i\in\mathcal{S}} H(X_i|Y)$ which satisfies condition (B) of Theorem~\ref{claim:erY}.

By our settings of group codes,
$X_\mathcal{S}$ is characterized by the trivial subgroup $\{\mathbf{i}\}$ ($G_\mathcal{S}=\{\mathbf{i}\}$). Note that $\mathcal{Y}$ is the set of cosets of $G'$ in $G$.
The function $f_Y$ takes as input $x_\mathcal{S}\in G$ (a coset of $\mathbf{i}$) and outputs $y$, the coset of $G'$ which contains $x_\mathcal{S}$. 
Such that, $A(y)=y$ and $A_i(y)$ is the set of cosets of $G'\cap G_i$ contained in $A(y)$.
By Lagrange's Theorem we have $|A_i(y)| = \frac{|G'|}{|G'\cap G_i|}$, $|\mathcal{X}_i|=\frac{|G|}{|G_i|}$ and $|\mathcal{X}_{e^*}|=\frac{|G|}{|G_{e^*}|}$. 
By the proof of Theorem 3 in \cite{wei2017effect}, we have 
\begin{equation*}
    \log \frac{|G'|}{|G_i\cap G'|}\geq \log \frac{|G_{e^*}|}{|G_i|}.    
\end{equation*}
Accordingly, $|A_i(y)|\geq |\mathcal{X}_i|/|\mathcal{X}_{e^*}|$. 
Since here $\epsilon=0$, condition (C) of Theorem~\ref{claim:erY} holds.

We have shown that for a network instance $\mathcal{I}$ which is $(\epsilon,\mathbf{R},n)$-feasible by Abelian group network codes, the local edge removal statement holds on every edge $e^*$ in $\mathcal{I}$ which implies the edge removal statement on $\mathcal{I}$.
\end{proof}


\section*{II. Proof of Corollary~\ref{claim:erXe}}

Here we prove that the local edge removal statement holds on edge $e^*$ when the edge random variable on $e^*$ satisfies certain conditions.  The proof of Corollary~\ref{claim:erXe} is limited to the $0$-error case. Roughly speaking, our proof does not naturally extend to the $\epsilon$-error case due to our use of an averaging argument on $A(x_{e^*}')$. To prove Corollary~\ref{claim:erXe} for $\epsilon$-error one could need the averaging argument to preserve error, which is not necessarily the case.



\begin{proof}[Proof of Corollary 3]
Here we replace the random variable $Y$ with $X_{e^*}$, since $X_{e^*}$ is a function of itself, the function $f_Y$ in the definition of $A(Y)$ and $A_i(Y)$ is the global encoding function $\phi_{ge^*}$.
Let $A(x_{e^*})=\{x_\mathcal{S}:x_{e^*}=\phi_{ge^*}(x_\mathcal{S})\}$, $A_i(x_{e^*})=\{x_i\in\mathcal{X}_i:\exists x_\mathcal{S}\in A(x_{e^*}), x_\mathcal{S}(i)=x_i\}$. Here $x_\mathcal{S}(i)$ denotes the $i$th entry of vector $x_\mathcal{S}$.

Set $Y=X_{e^*}$, \emph{Condition (A) and (B)} in Theorem~\ref{claim:erY} hold immediately. 
By an averaging argument, there exists $x_{e^*}'\in \mathcal{X}_{e^*}$ such that $|A(x_{e^*}')|\geq |\mathcal{X}_\mathcal{S}|/|\mathcal{X}_{e^*}|$. 
By definitions and the proof of Theorem~\ref{claim:erY}, we have  $|A(x_{e^*}')|=\prod_{i\in\mathcal{S}}|A_i(x_{e^*}')|$ and $|\mathcal{X}_\mathcal{S}|=\prod_{i\in\mathcal{S}} |\mathcal{X}_i|$, such that $|A_i(x_{e^*}')|\geq |\mathcal{X}_i|/|\mathcal{X}_{e^*}|$ for any $i\in\mathcal{S}$. 
Combined with the fact that $\epsilon=0$, the condition (C) in Theorem~\ref{claim:erY} holds, which suffice to prove our assertion.
\end{proof}


\section*{III. A claim regarding CWL functions}

\begin{claim}
Let $\mathcal{I}$ be a network coding instance. For edge $e\in\mathcal{E}$, if the global encoding function $\phi_{ge}$ on $e$ is CWL, then it is group characterizable (but not necessarily Abelian group characterizable).
\end{claim}

\begin{proof}
Recall that for CWL functions the alphabet of each source $i\in\mathcal{S}$ is a group $G_i$ and the edge message alphabet is a group $G_e$. 
Define $G'=G_1\times\dots\times G_{|\mathcal{S}|}$ and $G_i'=G_1\times\dots \{\mathbf{i}_i\}\dots \times G_{|\mathcal{S}|}$ for any $i\in\mathcal{S}$. 
Also, for $\alpha\subset \mathcal{S}$, let $G_\alpha'=\cap_{i\in\alpha} G_i'=\prod_{i\in \mathcal{S}\setminus\alpha} G_i \times \prod_{i\in \alpha}\{\mathbf{i}_i\}$.
By \cite{gallian2016contemporary}, we know that $\{G_\alpha':\alpha\subset \mathcal{S}\}$ are subgroups of $G'$.
In what follows we use ``$\cdot$'', ``$\circ$'' to denote the binary group operations on $G'$, $G_e$, respectively.


Here, and in what follows all products of the form $\prod_{i\in\alpha} g_i$ for $\alpha\subseteq \mathcal{S}$ are ordered by the natural ordering of $\mathcal{S}=\{1,\dots,|\mathcal{S}|\}$.

With the definition of group homomorphism $\phi_{ge}$ in Section~\ref{sec:CWL}, we define $G_e'=\phi_{ge}^{-1}(\mathbf{i}_e)$. 
We first show that $G_e'$ is a subgroup of $G'$.
For any $a,b\in G_e'$, $\phi_{ge}(a\cdot b)=\phi_{ge}(a)\circ \phi_{ge}(b)=\mathbf{i}_e$. 
Thus, $G_e'$ is closed. 
Let $\mathbf{i}'=(\mathbf{i}_i)_{i\in\mathcal{S}}\in G'$. For any $a\in G_e'$ we have $\phi_{ge}(a)=\phi_{ge}(\mathbf{i}'\cdot a) = \phi_{ge}(a\cdot\mathbf{i}')=\mathbf{i}_e\circ \phi_{ge}(\mathbf{i}')=\mathbf{i}_e$, therefore $\mathbf{i}'\in G_e'$ and $G_e'$ has an identity element.
Let $b\in G'$ and $b^{-1}$ be the inverse of $b$ in $G'$. By our definitions, $\phi_{ge}(\mathbf{i}')=\phi_{ge}(b\cdot b^{-1})=\phi_{ge}(b)\circ \phi_{ge}(b^{-1})=\mathbf{i}_e\circ\phi_{ge}(b^{-1})=\mathbf{i}_e$ which implies $\phi_{ge}(b^{-1})=\mathbf{i}_e$, $b^{-1}\in G_e'$, and $G_e'$ has inverse elements.
Thus $G_e'$ is a subgroup of $G'$.

\begin{lemma}\label{lemma:phiGe}
Let $\phi_{ge}$ be a group homomorphism from $G$ to $G_e$, $|\phi_{ge}^{-1}(g_e)|=|\phi_{ge}^{-1}(g_e')|$ for any $g_e\neq g_e'\in G_e$.
And $\{\phi_{ge}^{-1}(g_e):g_e\in G_e\}$ are cosets of $G_e'$.
\end{lemma}

\begin{proof}
For any $g_e\neq g_e'\in G_e$, let $g\in\phi_{ge}^{-1}(g_e)$ and $g'\in\phi_{ge}^{-1}(g_e')$.
Let $\delta\in G$ satisfy $g\cdot\delta= g'$, then $\phi_{ge}(g\cdot\delta) = \phi_{ge}(g)\circ\phi_{ge}(\delta) = g_e'$, which implies $\phi_{ge}(\delta) = g_e^{-1}\circ g_e'$. 
Now, for any $g''\in \phi_{ge}^{-1}(g_e)$, as $\phi_{ge}(g''\cdot \delta)= \phi_{ge}(g'')\circ \phi_{ge}(\delta) = g_e\circ g_e^{-1}\circ g_e' =g_e'$, we have $g''\cdot \delta\in \phi_{ge}^{-1}(g_e')$, such that $\{g''\cdot \delta:g''\in \phi_{ge}^{-1}(g_e)\}\subseteq \phi_{ge}^{-1}(g_e')$. Since $|\{g''\cdot \delta:g''\in \phi_{ge}^{-1}(g_e)\}|=|\phi_{ge}^{-1}(g_e)|$, we have $|\phi_{ge}^{-1}(g_e)|\leq|\phi_{ge}^{-1}(g_e')|$. 
Similarly, $|\phi_{ge}^{-1}(g_e')|\geq|\phi_{ge}^{-1}(g_e)|$, thus $|\phi_{ge}^{-1}(g_e')|=|\phi_{ge}^{-1}(g_e)|$. Since we have $|G_e|$ preimages sets, we have $|\phi_{ge}^{-1}(g_e)|=\frac{|G'|}{|G_e|}$ for any $g_e\in G_e$.
 
Let $g\in \phi^{-1}_{ge}(g_e)$ , let $g G_e'=\{g\cdot g_e':g_e'\in G_e'\}$.
For any $g'\in g G_e'$, $\phi_{ge}(g')=\phi_{ge}(g)\circ \mathbf{i}_e=g_e$.
Thus $g G_e'\subseteq \phi_{ge}^{-1}(g_e)$. 
By Lagrange's Theorem, we have $|g G_e'|=\frac{|G'|}{|G_e|}=|\phi_{ge}^{-1}(g_e)|$, such that $g G_e'= \phi_{ge}^{-1}(g_e)$. In other words, for any $g_e\in G_e$, $\phi_{ge}^{-1}(g_e)$ is a coset of $G_e'$ with
\begin{equation}\label{equa:cleq1}
    |\phi_{ge}^{-1}(g_e)|=|G_e'|.
\end{equation}
\end{proof}

To show that $\{G_\alpha':\alpha\subset \mathcal{S}\cup \{e\}\}$ is a group characterization of $\phi_{ge}$, we need to show that $H(X_\alpha)=\log\frac{|G'|}{|G_\alpha'|}$ for any $\alpha\subseteq\mathcal{S}\cup \{e\}$.
Here $G_\alpha'=\cap_{f\in\alpha} G_f'$.
 
For any $\alpha\subseteq \mathcal{S}$, recall that $G_\alpha=\prod_{f\in\alpha} G_f$ and in what follows we denote $\bar{\alpha}=\mathcal{S}\setminus\alpha$, $\mathbf{i}_\alpha=(\mathbf{i}_i)_{i\in\alpha}$.
By our definitions $G_{\bar{\alpha}}=\prod_{i\in\mathcal{S}\setminus\alpha} G_i$ and $G'=G_\alpha\times G_{\bar{\alpha}}$.

By our definitions, we have $|G_\alpha||G_\alpha'|=|G'|$ implying that 
\begin{equation*}
    H(X_\alpha)=\log |G_\alpha|=\log\frac{|G'|}{|G_\alpha'|}.
\end{equation*}
In addition $|G_e'|=\frac{|G'|}{|G_e|}$ implying that
\begin{equation}\label{equa:cleq01}
    H(X_e)=\log |G_e|=\log\frac{|G'|}{|G_e'|}.
\end{equation}
 
We now show that $X_{\alpha,e}$ is group characterizable, where $\alpha\subseteq \mathcal{S}$.
We first prove that
\begin{equation}\label{equa:cleq02}
\begin{aligned}
    H(X_\alpha|X_e)&=\sum_{g_e\in G_e} \Pr(X_e=g_e)H(X_\alpha|X_e=g_e)\\ &\stackrel{(a)}{=}H(X_\alpha|X_e=\mathbf{i}_e)\\ &\stackrel{(b)}{=}\log \frac{|\phi_{ge}^{-1}(\mathbf{i}_e)|} {|\phi_{ge}^{-1}(\mathbf{i}_e) \cap \{\mathbf{i}_{\alpha}\times G_{\bar{\alpha}}\}|}\\ &\stackrel{(c)}{=}\log \frac{|G_e'|}{|G_e'\cap G_\alpha'|}.
\end{aligned}
\end{equation}

By the uniformity of source random variables and by (\ref{equa:cleq1}), we have
\begin{equation}\label{equa:cleq2}
\begin{aligned}
       \Pr(X_e=g_e) =& \sum_{g_{\mathcal{S}}\in \phi_{ge}^{-1}(g_e)} \Pr(X_\mathcal{S}=g_{\mathcal{S}}) \\ =& \frac{|G_e'|}{|G'|} =\frac{1}{|G_e|}
\end{aligned}
\end{equation}
for any $g_e\in G_e$.

Let $\mathcal{Y}_i=\{g_e\in G_e: g_e=\phi_{ge}(g_i,\mathbf{i}_{\mathcal{S}\setminus i}), g_i\in G_i\}$.
For any $g_e\in\mathcal{Y}_i$, we define $A_i(g_e)=\{g_i\in G_i:\phi_{ge}(g_i,\mathbf{i}_{\mathcal{S}\setminus i})=g_e\}$ which is a subset of $G_i$. 
We partition $G_i$ into subsets $\{A_i (g_e): g_e\in \mathcal{Y}_i\}$. 
By Lemma~\ref{lemma:sameSize}, for any $g_{e,i}'\neq g_{e,i}''\in \mathcal{Y}_i$, we have $|A_i(g_{e,i}')|=|A_i(g_{e,i}'')|$.
For any $\alpha\subseteq \mathcal{S}$, let 
\begin{equation*}
    \mathcal{Y}_\alpha= \left\{ \prod_{i\in\alpha} g_{e,i}: g_{e,i}\in \mathcal{Y}_i, i\in\alpha \right\}.
\end{equation*}
Note that $\mathcal{Y}_\alpha\subseteq G_e$.
For any $g_{e,\alpha} \in\mathcal{Y}_\alpha$, we define
\begin{equation*}
    A_\alpha(g_{e,\alpha})=\left\{\prod_{i\in\alpha} A_i(g_{e,i}):\prod_{i\in\alpha} g_{e,i}=g_{e,\alpha}\right\}.
\end{equation*}
As $A_i(g_{e,i})\subset G_i$ for each $i\in\alpha$, $A_\alpha(g_{e,\alpha})$ is a subset of $G_\alpha=\prod_{i\in\alpha} G_i$.

\begin{lemma}\label{lemma:inverseInY_a}
For any $\alpha\subseteq \mathcal{S}$ and $g_{e,\alpha}\in \mathcal{Y}_\alpha$, we have $g_{e,\alpha}^{-1}\in \mathcal{Y}_\alpha$.
\end{lemma}
\begin{proof}
As $\phi_{ge}(\mathbf{i}_\mathcal{S})=\phi_{ge}(\mathbf{i}_\mathcal{S}\cdot \mathbf{i}_\mathcal{S})=\phi_{ge}(\mathbf{i}_\mathcal{S})\circ \phi_{ge}(\mathbf{i}_\mathcal{S})$, we have 
\begin{equation} \label{equa:identity}
    \phi_{ge}(\mathbf{i}_\mathcal{S})=\mathbf{i}_e.
\end{equation}
Let $g_{e,\alpha}\in\mathcal{Y}_\alpha$ and $g_\alpha\in G_\alpha$, such that $\phi_{ge}(g_\alpha,\mathbf{i}_{\bar{\alpha}})=g_{e,\alpha}$.
By our definition $g_\alpha=(g_i)_{i\in\alpha}$, define $g_\alpha^{-1}=(g_i^{-1})_{i\in\alpha}$ where $g_i\in G_i$.
We have 
\begin{equation*}
\begin{aligned}
   \phi_{ge}(\mathbf{i}_\mathcal{S}) =&\phi_{ge}(g_\alpha\cdot g_\alpha^{-1},\mathbf{i}_{\bar{\alpha}})\\ =& \phi_{ge}(g_\alpha,\mathbf{i}_{\bar{\alpha}}) \circ \phi_{ge}(g_\alpha^{-1},\mathbf{i}_{\bar{\alpha}})\\ =& g_{e,\alpha}\circ \phi_{ge}(g_\alpha^{-1},\mathbf{i}_{\bar{\alpha}}).
\end{aligned}
\end{equation*}
Thus $\phi_{ge}(g_\alpha^{-1},\mathbf{i}_{\bar{\alpha}}) = \prod_{i\in\alpha}\phi_{ge}(g_i^{-1},\mathbf{i}_{\mathcal{S}\setminus i})=g_{e,\alpha}^{-1}$ implying $g_{e,\alpha}^{-1}\in \mathcal{Y}_\alpha$.
\end{proof}

Let $g_{e,\alpha}'\neq g_{e,\alpha}''\in\mathcal{Y}_\alpha$.
By Lemma~\ref{lemma:inverseInY_a} we have $g_{e,\alpha}'^{-1}\in\mathcal{Y}_\alpha$.
Let $g_\alpha', g_\alpha''\in G_\alpha$, such that $\phi_{ge}(g_\alpha',\mathbf{i}_{\bar{\alpha}})=g_{e,\alpha}'^{-1}$ and $\phi_{ge}(g_\alpha'',\mathbf{i}_{\bar{\alpha}})=g_{e,\alpha}''$.
Define $\delta=g_\alpha'\cdot g_\alpha''$, such that 
\begin{equation*}
   \phi_{ge}(\delta,\mathbf{i}_{\bar{\alpha}})= \phi_{ge}(g_\alpha',\mathbf{i}_{\bar{\alpha}})\circ \phi_{ge}(g_\alpha'',\mathbf{i}_{\bar{\alpha}})=g_{e,\alpha}'^{-1}\circ g_{e,\alpha}''.
\end{equation*}

Let $g_\alpha\in A_\alpha(g_{e,\alpha}')$. 
As $\phi_{ge}(g_\alpha\cdot\delta ,\mathbf{i}_{\bar{\alpha}}) = \phi_{ge}(g_\alpha,\mathbf{i}_{\bar{\alpha}}) \circ \phi_{ge}(\delta,\mathbf{i}_{\bar{\alpha}}) = g_{e,\alpha}'\circ g_{e,\alpha}'^{-1} \circ g_{e,\alpha}''= g_{e,\alpha}''$, we have $\{g_\alpha \cdot \delta:g_\alpha\in G_\alpha\}\subseteq A_\alpha(g_{e,\alpha}'')$. 
As $|A_\alpha(g_{e,\alpha}')|=|\{g_\alpha\cdot \delta:g_\alpha\in A_\alpha(g_{e,\alpha}')\}|$, we have $|A_\alpha(g_{e,\alpha}')|\leq |A_\alpha(g_{e,\alpha}'')|$.
Similarly, we can prove $|A_\alpha(g_{e,\alpha}')|\geq |A_\alpha(g_{e,\alpha}'')|$.
Thus, for $\alpha\subseteq \mathcal{S}$,
\begin{equation}\label{equa:samesizei}
    |A_\alpha(g_{e,\alpha}')| = |A_\alpha(g_{e,\alpha}'')|
\end{equation}
for any $g_{e,\alpha}'\neq g_{e,\alpha}''\in \mathcal{Y}_\alpha$.

By the definitions above, we can partition $G'$ into cells. 
Namely, for any $\alpha\subseteq \mathcal{S}$ and $\bar{\alpha}=\mathcal{S}\setminus\alpha$, each cell is a subset $A_\alpha(g_{e,\alpha})\times A_{\bar{\alpha}}(g_{e,\bar{\alpha}})$, where $g_{e,\alpha}\in \mathcal{Y}_\alpha$ and $g_{e,\bar{\alpha}}\in \mathcal{Y}_{\bar{\alpha}}$.

\begin{lemma}\label{lemma:diff}
Let $g_e\in G_e$.
Let $A_\alpha(g_{e,\alpha}')\times A_{\bar{\alpha}}(g_{e,\bar{\alpha}}')$ and $A_\alpha(g_{e,\alpha}'')\times A_{\bar{\alpha}}(g_{e,\bar{\alpha}}'')$ be two disjoint cells in $\phi_{ge}^{-1}(g_e)$. Let $g_{e,\alpha}',g_{e,\alpha}''\in \mathcal{Y}_\alpha$, $g_{e,\bar{\alpha}}',g_{e,\bar{\alpha}}''\in\mathcal{Y}_{\bar{\alpha}}$, such that $(g_{e,\alpha}', g_{e,\bar{\alpha}}')\neq (g_{e,\alpha}'', g_{e,\bar{\alpha}}'')$ and $g_{e,\alpha}'\circ g_{e,\bar{\alpha}}'=g_{e,\alpha}''\circ g_{e,\bar{\alpha}}''=g_e$.
We have that $g_{e,\alpha}'\neq g_{e,\alpha}''$ and $g_{e,\bar{\alpha}}'\neq g_{e,\bar{\alpha}}''$.
\end{lemma}

\begin{proof}

\begin{figure}
    \centering
    \includegraphics[draft=false,width=0.5\textwidth]{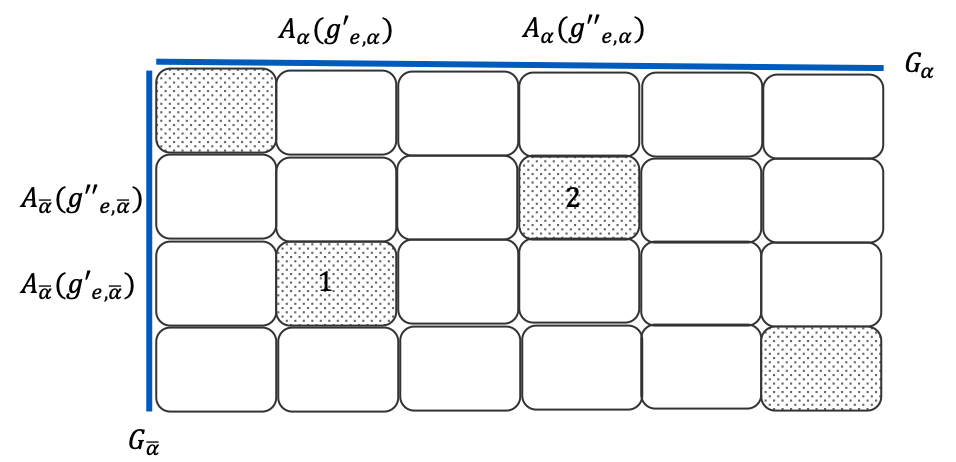}
    \caption{An example depicting the proof of Lemma~\ref{lemma:diff}. The dark blocks represent the set $\phi_{ge}^{-1}(g_e)$. Block 1 represents $A_\alpha(g_{e,\alpha}')\times A_{\bar{\alpha}}(g_{e,\bar{\alpha}}')$ and Block 2 represents $A_\alpha(g_{e,\alpha}'')\times A_{\bar{\alpha}}(g_{e,\bar{\alpha}}'')$.}
    \label{fig:lemma}
\end{figure}

The proof is depicted in Figure~\ref{fig:lemma}. 

Assume $g_{e,\alpha}'=g_{e,\alpha}''$, as $g_{e,\alpha}'\circ g_{e,\bar{\alpha}}'=g_{e,\alpha}''\circ g_{e,\bar{\alpha}}''=g_e$, we have $g_{e,\bar{\alpha}}'=g_{e,\bar{\alpha}}''$. 
This contradicts the assumption that $(g_{e,\alpha}', g_{e,\bar{\alpha}}')\neq (g_{e,\alpha}'', g_{e,\bar{\alpha}}'')$.
\end{proof}

We now show that for any cell $A_\alpha(g_{e,\alpha})\times A_{\bar{\alpha}}(g_{e,\bar{\alpha}})$, 
\begin{equation} \label{equa:projalpha}
\begin{aligned}
    |A_{\bar{\alpha}}(g_{e,\bar{\alpha}})| &\stackrel{(a)}{=} |A_{\bar{\alpha}}((\mathbf{i}_e))|\\  &\stackrel{(b)}{=}|\phi_{ge}^{-1}(\mathbf{i}_e) \cap \{\mathbf{i}_{\alpha}\times G_{\bar{\alpha}}\}|.
\end{aligned}
\end{equation}

By (\ref{equa:identity}) we have $\phi_{ge}(\mathbf{i}_i,\mathbf{i}_{\mathcal{S}\setminus i})=\mathbf{i}_e$.
As $\mathbf{i}_i\in G_i$, by the definition of $\mathcal{Y}_i$, we know that $\mathbf{i}_e\in \mathcal{Y}_i$. 
With the definition of $\mathcal{Y}_\alpha$, we have $\prod_{i\in\alpha}\mathbf{i}_e=\mathbf{i}_e\in \mathcal{Y}_\alpha$.
By (\ref{equa:samesizei}), equation (a) of (\ref{equa:projalpha}) holds.
With our definitions, for any $g_e\in G_e$, we have
\begin{equation}\label{equa:phi_ge'}
    \phi_{ge}^{-1}(g_e)= \{A_\alpha(g_{e,\alpha})\times A_{\bar{\alpha}}(g_{e,\bar{\alpha}}): g_{e,\alpha}\circ g_{e,\bar{\alpha}} =g_e\}.
\end{equation}

As $\mathbf{i}_e\in \mathcal{Y}_\alpha$ for any $\alpha\subseteq \mathcal{S}$ and as $\mathbf{i}_e\circ\mathbf{i}_e=\mathbf{i}_e$, by (\ref{equa:phi_ge'}) we have $A_\alpha(\mathbf{i}_e)\times A_{\bar{\alpha}}(\mathbf{i}_e)\subset \phi_{ge}^{-1}(\mathbf{i}_e)$. 
By Lemma~\ref{lemma:diff}, $A_\alpha(\mathbf{i}_e)\times A_{\bar{\alpha}}(\mathbf{i}_e)$ is the only cell in $\phi_{ge}^{-1}(\mathbf{i}_e)$ with $g_{e,\alpha}=\mathbf{i}_e\in\mathcal{Y}_\alpha$ or  $g_{e,\bar{\alpha}}=\mathbf{i}_e\in\mathcal{Y}_{\bar{\alpha}}$. 
The size of $A_{\bar{\alpha}}(\mathbf{i}_e)$ is exactly the projection of the cell $A_\alpha(\mathbf{i}_e)\times A_{\bar{\alpha}}(\mathbf{i}_e)$ onto $G_{\bar{\alpha}}$.
Which by Lemma~\ref{lemma:diff} is exactly the intersection of $(\mathbf{i}_\alpha\times G_{\bar{\alpha}})$ and $\phi_{ge}^{-1}(\mathbf{i}_e)$.
Namely, $|A_{\bar{\alpha}}(\mathbf{i}_e)| = |\phi_{ge}^{-1}(\mathbf{i}_e) \cap \{ \mathbf{i}_{\alpha} \times A_{\bar{\alpha}}(\mathbf{i}_e)\}| = |\phi_{ge}^{-1}(\mathbf{i}_e) \cap \{\mathbf{i}_{\alpha} \times G_{\bar{\alpha}}\}|$, and the equation (b) of (\ref{equa:projalpha}) holds.
 
The projection of $\phi_{ge}^{-1}(g_e)$ on $G_\alpha$ is defined as $G_\alpha(g_e)=\{g_\alpha: \exists (g_\alpha, g_{\bar{\alpha}})\in \phi_{ge}^{-1}(g_e), g_{\bar{\alpha}}\in G_{\bar{\alpha}}\}$.
By (\ref{equa:cleq1}), (\ref{equa:projalpha}) and (\ref{equa:phi_ge'}),
\begin{equation*}
    |G_\alpha(g_e)| = \frac{|\phi_{ge}^{-1}(g_e)|}{|A_{\bar{\alpha}}(\mathbf{i}_e)|} = \frac{|\phi_{ge}^{-1}(\mathbf{i}_e)|} {|\phi_{ge}^{-1}(\mathbf{i}_e) \cap \{\mathbf{i}_{\alpha} \times G_{\bar{\alpha}}\}|}.
\end{equation*}
Conditioning on $X_e=g_e$, as $X_\alpha$ is uniformly distributed over $G_\alpha(g_e)$, for any $g_e\in G_e$ we have
\begin{equation*} 
    \Pr(X_\alpha=g_\alpha|X_e=g_e) = \frac{1}{|G_\alpha(g_e)|} = \frac{|\phi_{ge}^{-1}(\mathbf{i}_e) \cap \{ \mathbf{i}_{\alpha}\times G_{\bar{\alpha}}\}|}{|\phi_{ge}^{-1}(\mathbf{i}_e)|}
\end{equation*}
for $g_\alpha\in G_\alpha(g_e)$ and $\Pr(X_\alpha=g_\alpha|X_e=g_e)=0$ for $g_\alpha\in G_\alpha \setminus G_\alpha(g_e)$.
Namely, 
\begin{equation}\label{equa:cleq3}
    H(X_\alpha|X_e=g_e)= \log \frac{|\phi_{ge}^{-1}(\mathbf{i}_e)|}{|\phi_{ge}^{-1}(\mathbf{i}_e) \cap \{\mathbf{i}_{\alpha} \times G_{\bar{\alpha}}\}|}   
\end{equation}
for any $g_e\in G_e$.
Accordingly, with (\ref{equa:cleq2}) and (\ref{equa:cleq3}), equation (a) of (\ref{equa:cleq02}) holds and with (\ref{equa:cleq3}) equation (b) of (\ref{equa:cleq02}) holds. By our definition of $G_e'$ and $G_\alpha'$, equation (c) of (\ref{equa:cleq02}) holds. Thus, by (\ref{equa:cleq01}) and (\ref{equa:cleq02})
\begin{equation*}
    H(X_e)+H(X_\alpha|X_e)=\log\frac{|G'|}{|G_e'\cap G_\alpha'|}=H(X_e,X_\alpha).
\end{equation*}
 
We conclude that the group $G'$ with subgroups $\{G_f':f\in\mathcal{S}\cup\{e\}\}$ is a group characterization of random variables $\{X_f:f\in\mathcal{S}\cup e\}$. Accordingly $\phi_{ge}$ is group characterizable \cite{chan2008dualities}. Note that, since we do not assume groups $\{G_f:f\in\mathcal{S}\cup\{e\}\}$ are Abelian, $G'$ is not necessarily Ablian.
\end{proof}


\section*{IV. Proof of Theorem~\ref{claim:ER_piecelinear}}

\begin{proof}
By assumption, $\phi_{ge^*}$ is $K$-piece-wise CWL with $K$ subfunctions $\{\phi_{ge^*}^{(k)}:G_\mathcal{S}\mapsto G_{e^*}\}_{k\in[K]}$.
By an averaging argument, there exists some $k'\in [K]$ corresponding to the subfunction $\phi_{ge^*}^{(k')}$ for which $|S^{(k')}|\geq |G_\mathcal{S}|/K$. 
For every $i\in\mathcal{S}$
\begin{equation}\label{equa:SGM}
    |S_i^{(k')}|\geq \frac{|G_i|}{K},
\end{equation}
as otherwise $\prod_{i\in\mathcal{S}} |S_i^{(k')}| < {\prod_{i\in\mathcal{S}} |G_i|}/{K}={|G_\mathcal{S}|}/{K}$.

For any $i\in\mathcal{S}$, let $G_{e^*}(i) = \{\phi^{(k')}_{ge^*}(x_i,\mathbf{i}_{\mathcal{S}\setminus i}):x_i\in G_i\}$.
For any $x_{e^*}\in G_{e^*}(i)$, we define $A_i(x_{e^*})=\{x_i\in G_i:\phi^{(k')}_{ge^*}(x_i, \mathbf{i}_{\mathcal{S}\setminus i})=x_{e^*}\}$.
The set $A_i(x_{e^*})$ is defined similarly to $A_i(y)$ in the proof of Theorem~\ref{claim:ER_linear}. By Lemma~\ref{lemma:sameSize} for any $x_{e^*},\bar{x}_{e^*}\in G_{e^*}(i)$, we have $|A_i(x_{e^*})|=|A_i(\bar{x}_{e^*})|$.
For any $i\in\mathcal{S}$, we call the elements in $S^{(k')}_i$ ``good'' and elements in $G_i\setminus S^{(k')}_i$ ``bad''.
By an averaging argument there exists $x_{e^*}(i)\in G_{e^*}(i)$, such that $A_i(x_{e^*}(i))$ contains at least a fraction of $\frac{|S^{(k')}_i|}{|G_i|}$ ``good'' elements. 
We define $G_i'=\{x_i\in A_i(x_{e^*}(i)): \mbox{$x_i$ is good}\}$, accordingly 
\begin{equation} \label{equa:AGood}
    |G_i'|\geq \frac{|S^{(k')}_i|}{|G_i|}|A_i(x_{e^*}(i))|.
\end{equation}
We use $\prod_{i\in\mathcal{S}}G_i'$ to construct a new code $\{X_f':f\in\mathcal{S}\cup\mathcal{E}\}$. Here, $X_i'$ is uniformly distributed over $G_i'$, such that $H(X_i')=n R_i'= \log |G_i'|$.
For the new code, sources are independent since there is no cooperation during transmission.
As the supports of new source random variables are subsets of the original ones ($G_i'\subset G_i$ for any $i\in\mathcal{S}$), we apply the original encoding and decoding functions to the new code.
As $\phi_{ge^*}^{(k')}$ is CWL, the message $x_{e^*} = \phi_{ge^*}^{(k')}(g)$ for any $g \in \prod_{i \in \mathcal{S}}G_i'$ transmitted through $e^*$ equals $\stackrel[i\in\mathcal{S}]{e^*}{\circ} x_{e^*}{(i)}$ which is constant. 
Thus no information needs to be transmitted through edge $e^*$ and we may consider it removed.
As $|G_{e^*}(i)|\leq |G_{e^*}|$, we know $|A_i(x_{e^*}(i))|=|G_i|/|G_{e^*}(i)|\geq |G_i|/|G_{e^*}|$, with (\ref{equa:SGM}) and (\ref{equa:AGood}), we have 
\begin{equation}\label{equa:lastone}
    |G_i'|\geq \frac{|S^{(k')}_i|}{|G_i|}|A_i(x_{e^*}(i))|\geq \frac{|A_i(x_{e^*}(i))|}{K}= \frac{|G_i|}{|G_{e^*}| \cdot K}.
\end{equation}
Since $H(X_i')=n R_i'=\log |G_i'|$, $H(X_i)=n R_i= \log |G_i|$ and $H(X_{e^*})\leq n R_{e^*}=\log |G_{e^*}|$, by (\ref{equa:lastone}) we have
\begin{equation*}
    R_i'\geq R_i-R_{e^*}-\frac{\log K}{n}.
\end{equation*}
\end{proof}

\textbf{Remarks:} As we can see, the proof of Theorem~\ref{claim:ER_piecelinear} is limited to the $0$-error case. Roughly speaking, our proof does not naturally extend to the $\epsilon$-error case due to our use of an averaging argument on $\{1\dots k\}$. To prove Theorem~\ref{claim:ER_piecelinear} for $\epsilon$-error one could need the averaging argument to preserve error, which is not necessarily the case.


\section*{V. Case studies of Statement~\ref{statement:task}}

\subsection{The network instance of \cite{dougherty2005insufficiency}}

The network instance of \cite{dougherty2005insufficiency}, referred to as $N_3$, includes three subnetworks, referred here as networks $A$, $B$ and $C$ (from left to right). In \cite{dougherty2005insufficiency} it is shown that $N_3$ is solvable but not linearly solvable. In Fig.~\ref{fig:network_insuff} the original zero-error blocklength $n=1$ non-linear solution for $N_3$ of  \cite{dougherty2005insufficiency} is presented. For all edges $e^*$ in networks $A$ and $B$, the original code is linear. We thus focus on the validity of the CWL statement when $e^*$ is taken from network $C$. Toward this end, we suggest a modified code for which all edges in network $C$ are linear (at the ``price'' of modifying the encoding on edges of network $B$ to be non-linear).
Our modified code preserves the original encoding presented in Fig.~\ref{fig:network_insuff} on all edges of network $A$. Using the notation of \cite{dougherty2005insufficiency} for network $B$ we set the encoding functions to be:
\begin{equation*}
    \begin{aligned}
        e_{15,19}&=a+b+t(c)\\ e_{8,23}&=a+b\\ e_{24,32}&=a+t(c)\\ e_{25,33}&=b+t(c)\\
    \end{aligned}
\end{equation*}
For network $C$ we set the linear encoding functions to be:
\begin{equation*}
    \begin{aligned}
        e_{16,20}&=c+d+e\\ e_{26,34}&=c+d\\ e_{27,35}&=c+e\\ e_{28,36}&=d+e\\
    \end{aligned}
\end{equation*}
It now follows that the demands are met:
\begin{equation*}
    \begin{aligned}
        n_{40}: c=&t(e_{19,40}-e_{31,40})=t(a+b+t(c)-(a+b))\\ n_{41}: b=&e_{19,41}-e_{32,41}=a+b+t(c)-(a+t(c))\\ n_{42}: a=&e_{19,42}-e_{33,42}=a+b+t(c)-(b+t(c))\\ n_{43}: c=&(e_{33,43}+e_{32,43}-e_{31,43})\\ &+t(e_{34,43}+e_{35,43}-e_{36,43})\\ =&t(a+t(c)+b+t(c)-(a+b))\\ &+t(c+d+c+e-(d+e))\\ =&t(2c)+2t(c)\\ n_{44}: e=&e_{20,44}-e_{34,44}=(c+d+e)-(c+d)\\ n_{45}: d=&e_{20,45}-e_{35,45}=(c+d+e)-(c+e)\\n_{46}: c=&e_{20,46}-e_{36,46}=(c+d+e)-(d+e)\\
    \end{aligned}
\end{equation*}

\begin{figure}    \label{fig:network_insuff}
    \centering
    \includegraphics[draft=false,width=0.5\textwidth]{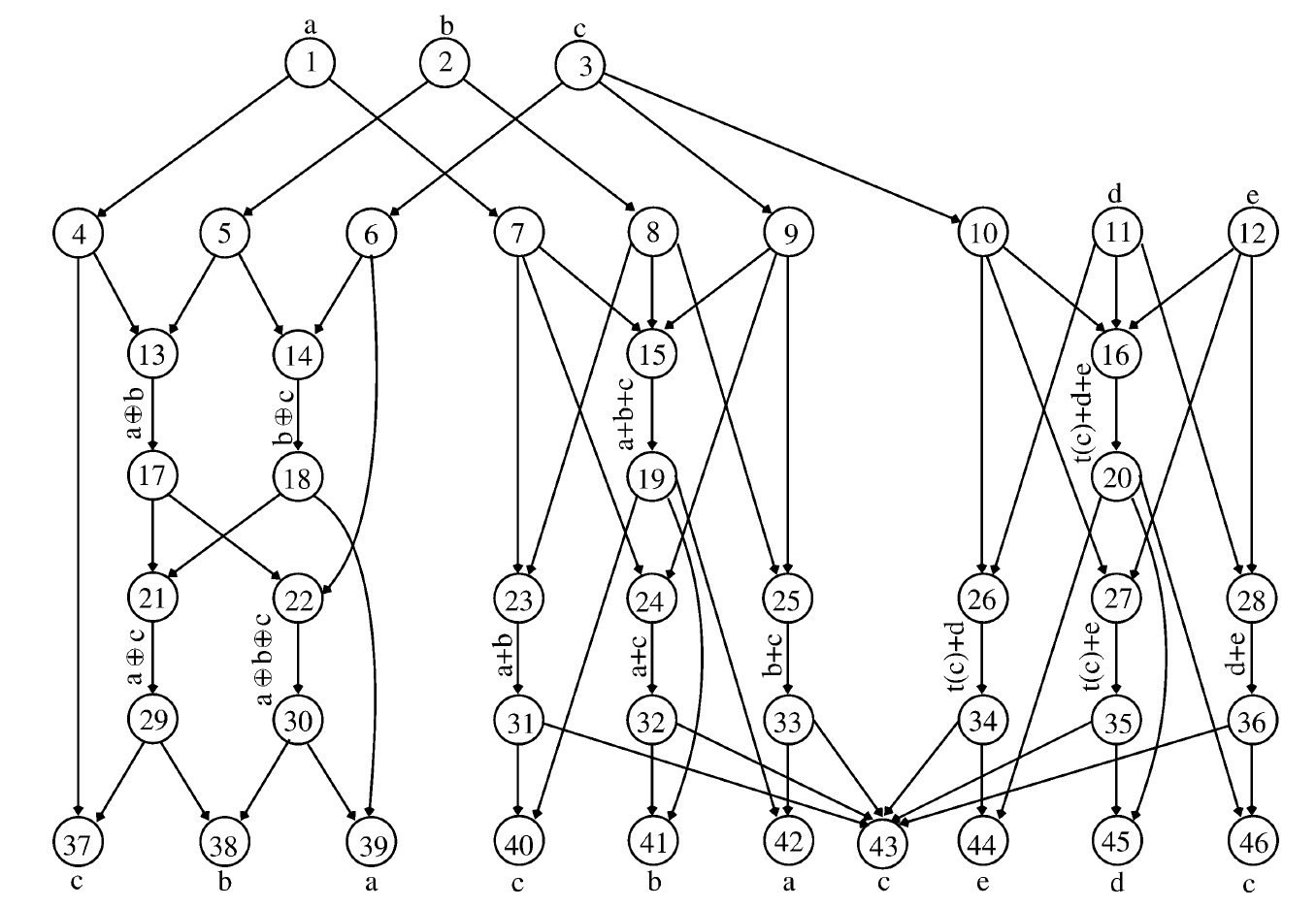}
    \caption{The network $N_3$ and coding scheme from \cite{dougherty2005insufficiency}. Figure taken from \cite{dougherty2005insufficiency}.}

\end{figure}


\subsection{The network instance of \cite{chan2008dualities}}

In \cite{chan2008dualities}, a duality between  entropy functions and network codes is introduced. 
For the network instance $N$ given in Figure~\ref{fig:network_duality} and for a given rate vector $\mathbf{R}$ (induced by an entropy vector which violates the Ingleton inequality) that is not linearly feasible, \cite{chan2008dualities} present a $(0,\mathbf{R},n)$-feasible non-linear code.

\begin{figure}\label{fig:network_duality}
    \centering
    \includegraphics[draft=false,width=0.45\textwidth]{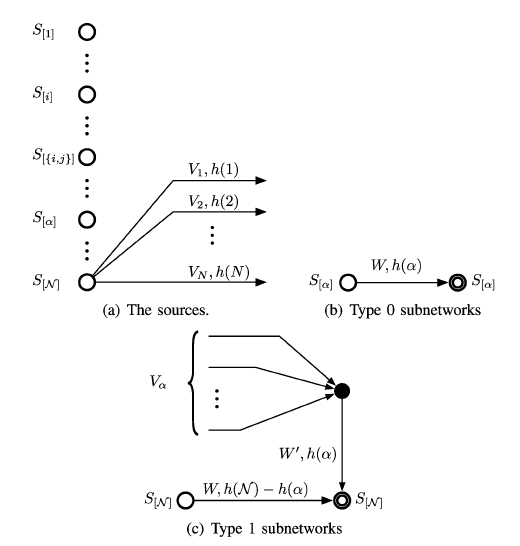}
    \includegraphics[draft=false,width=0.45\textwidth]{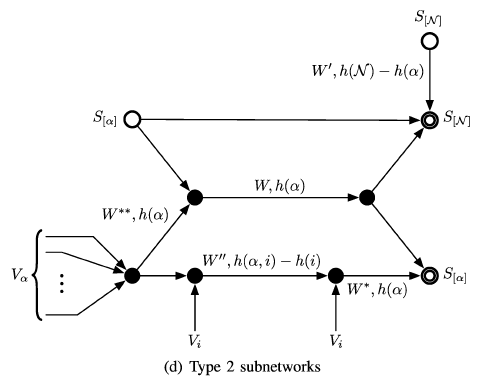}
    \caption{The network $N$ from \cite{chan2008dualities} with corresponding notation used in our analysis. The Figure is taken from \cite{chan2008dualities}.}
    
\end{figure}

Using the notation from  \cite{chan2008dualities}, let $\mathcal{N}=\{1,2,3,4\}$ and consider an entropy function $h$ of four quasi-uniform random variables $U_1, U_2, U_3, U_4$, such that
\begin{equation*}
    \begin{aligned}
    h(1)&= h(2)=h(3)=h(4)=\log 13\\ h(1,2)&= \log 6+\log 13\\ h(3,4)&= \log 13+\log 12\\ h(1,3)&= h(1,4) = h(2,3) = h(2,4) = \log 13+\log 4\\ h(i,j,k)&=\log 13 +\log 12\\ &=h(1,2,3,4),\ \forall \mbox{distinct}\ i,j,k.\\
    \end{aligned}
\end{equation*}
It is easy to show that $h$ violates the Ingleton inequality. \cite{chan2008dualities} show that the rate vector $\mathbf{R}=(R_\alpha)_{\alpha\subseteq \mathcal{S}}$ induced from $h$ is achievable on $N$ (with a non-linear code), where $R_\alpha=h(\alpha)$ is the rate of source $S[\alpha]$ for $\alpha\subseteq \mathcal{N}$. For every $e^*$ in $N$, we modify the code of \cite{chan2008dualities}, such that (a) it remains feasible, and (b) the encoding function on $e^*$ is CWL. We address edges of each subnetwork specified in Figure~\ref{fig:network_duality} separately. Before we start, we prove the following lemma used in our analysis.

\begin{lemma}\label{lemma:linear}
Let $A$, $B$ be two sets of size $p,q$, respectively for which $q$ divides $p$. Let $g$ be a balanced function that maps $A$ to $B$, i.e., the pre-image of every $b$ under $g$ is of size $p/q$. Then there exists a CWL function $f$ that maps $A$ onto $B$ that equals $g$.
\end{lemma}

\begin{proof}
Since $q$ divides $p$, we have $p=k q$, here $k$ is a positive integer.
First, we relabel elements in $A$, $B$ as elements in $G_k\times G_q$, $G_q$, respectively. Here $G_k=\{0,1,\dots,k-1\}$ and $G_q=\{0,1,\dots,q-1\}$ are the additive groups with addition modulo $k$ and $q$ respectively.
The relabeling is done according to function $g$. Namely, for each $j \in B$, let $A_j$ be the pre-image of $j$ under $g$. By our assumptions, $A_j$ is of size $k$. We label element $i$ of $A_j$ by $(i,j) \in G_k\times G_q$.
Now, define $f:A\mapsto B$ by $f(i,j)=j$.
Note that $f$ is equal to $g$.
For any two elements $a=(i,j)$, $a'=(i',j')$ of $A$ we have $f(a+a')=j+j'=f(a)+f(a')$ which implies that $f$ is CWL.
\end{proof}


\emph{Subnetwork (a): The sources.}
We study the outgoing edge carrying message $V_i$ for each $i\in\mathcal{N}$. Using the definitions in \cite{chan2008dualities} we have $H(S_{[\mathcal{N}]})=h(\mathcal{N})$, $H(V_i)=h(i)$ and $S_{[\mathcal{N}]}$, $V_i$ are quasi-uniform. This implies that the supports of $\mathcal{S}_{[\mathcal{N}]}$ and $V_i$ are of sizes $13\times12$ and $13$ respectively. We relabel the supports $\mathcal{S}_{[\mathcal{N}]}$, $\mathcal{V}_i$ as groups of size $13\times12$, $13$, respectively. By Lemma~\ref{lemma:linear}, once we relabel the elements in $\mathcal{S}_{[\mathcal{N}]}$, $\mathcal{V}_i$ the global encoding function corresponding to message $V_i$ is CWL.

\emph{Subnetwork (b): Type 0 subnetworks.}
For the edge carrying message $W$, the encoding function is CWL since $W$ is uncoded and equals $S_{[\alpha]}$. 

\emph{Subnetwork (c): Type 1 subnetworks.}
The edge carrying $W'$ is of capacity $h(\alpha)$. As $2^{h(\alpha)}$ divides $2^{h(\mathcal{N})}$ for any $\alpha\subseteq \mathcal{N}$, by Lemma~\ref{lemma:linear} the quasi-uniform $V_\alpha$ can be represented as a CWL function of $\mathcal{S}_{[\mathcal{N}]}$. 

For the edge carrying $W$, we note that $2^{h(\mathcal{N})-h(\alpha)}$ divides $2^{h(\mathcal{N})}$ for any  $\alpha\subset \mathcal{N}$ and that the function corresponding to $W$ is balanced in the solution of \cite{chan2008dualities}. Thus, by Lemma~\ref{lemma:linear}, there exists a coding scheme such that $W$ is a CWL function of source message $\mathcal{S}_{[\mathcal{N}]}$.

\emph{Subnetwork (d): Type 2 subnetworks.}
Using the same argument as in Type 1 subnetworks, one can represent $W'$, $W^*$ and $W^{**}$ as CWL functions. For the edge carrying $W$, following the analysis on Type 1 networks, $V_\alpha$ can be represented as a CWL function of $\mathcal{S}_{[\mathcal{N}]}$. Let $G_\alpha$ be the group corresponding to $V_\alpha$. Since $H(\mathcal{S}_{[\alpha]})=H(V_\alpha)=h(\alpha)$, we may relabel the support of $\mathcal{S}_{[\alpha]}$ as the group $G_\alpha$ of size $2^{h(\alpha)}$. Then the function that takes the pair $\mathcal{S}_{[\mathcal{N}]}$ and $\mathcal{S}_{[\alpha]}$ and returns the corresponding combination $V_\alpha \circ \mathcal{S}_{[\alpha]}$, where $\circ$ is the group operation of $G_\alpha$, is CWL.

For $W''$, similar to the analysis for $W$ in Type 1 networks, we have that $W''$ of \cite{chan2008dualities} is balanced and thus can be represented as a CWL function.


\subsection{The network instance of \cite{connelly2017class}}

\begin{figure}
    \centering
    \includegraphics[draft=false,width=0.35\textwidth]{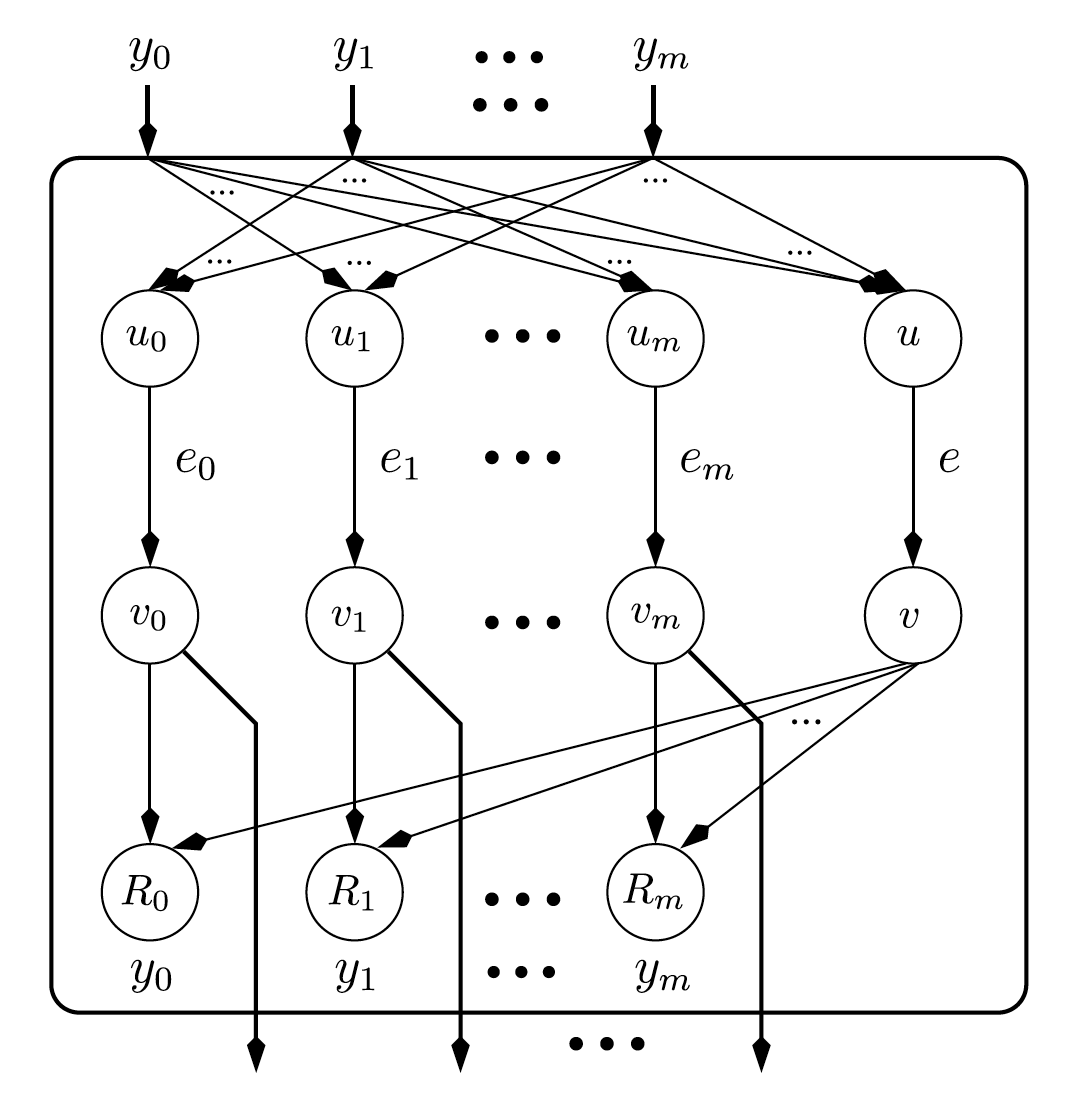}
    \caption{The building block $B(m)$ from \cite{connelly2017class} with the corresponding notation used in our analysis. The figure is taken from \cite{connelly2017class}.}
    \label{fig:B(m)}
\end{figure}

In \cite{connelly2017class}, a class of networks which is solvable but not linearly solvable is introduced. 
Using the notation of \cite{connelly2017class}, a non-linearly solvable network $N_4(m)$ is constructed by taking a disjoint union of component networks $N_1(m)$, $N_2(m,w)$, $N_3(m_1,m_2)$ with carefully chosen parameters. The later networks, in turn, are constructed from a certain building block $B(m)$, see Fig.~\ref{fig:B(m)}.

In $B(m)$ the only encoding edges are the edges carrying message $e_0,e_1,\dots,e_m, e$, all the other edges forward messages without encoding. Also, the edges carrying $e_0,e_1,\dots,e_m, e$ are directly connected to sources such that the local and global encoding function on these edges are identical (as the incoming edge messages are the source messages). Now we show, for any edge in the network $N_4$ defined in \cite{connelly2017class} that we can modify the code of \cite{connelly2017class} and make the global encoding function $\phi_{ge}$ CWL.

\subsubsection{Subnetwork: $N_1(m)$}
In the coding scheme of \cite{connelly2017class} all edges in $N_1(m)$ have a linear global encoding function.

\subsubsection{Subnetwork: $N_2(m,w)$}
By Lemma IV.4 of \cite{connelly2017class}, a solution for $N_2(m,w)$ is given as a single blocklength code over the ring $\mathbf{Z}_{mw}$ for each $l=1,2,\dots,w$ by 
\begin{equation*}
\begin{aligned}
   e_0^{(l)} &= \stackrel[j=1]{m+1}{\sum}x_j^{(l)}, \\e_i^{(l)} &= \pi_l(z)+\stackrel[j=1,j\neq i]{m+1}{\sum}x_j^{(l)}\ (i=1,2,\dots,m+1),\\ e^{(l)} &= \pi_l(z)+\stackrel[j=1]{m+1}{\sum}x_j^{(l)}.
\end{aligned}
\end{equation*}
From the expressions we notice that the edge message $e_0^{(l)}$ is a linear function of source messages.
By the proof of Lemma IV.2 of \cite{connelly2017class}, for each integer $y$, there exist unique integers $q_y$, $r_y$ such that $y=q_y m+r_y$ and $0\leq r_y< m$. 
For each $a=q_a m+r_a\in \mathbf{Z}_{mw}$ s.t. $0\leq r_a< m$, let $\hat{r}_a$ be the unique integer in $\{0,1,\dots,m-1\}$ such that $\hat{r}_a=r_a+1(\mod\ m)$ and for each $l=1,2,\dots, w-1$, define a permutation of $\mathbf{Z}_{mw}$ as follows:
\begin{equation*}
\begin{aligned}
   \pi_l(a)=&\begin{cases}
   q_a m+\hat{r}_a,& if\ q_a=l\\
   q_a m+r_a,& otherwise
   \end{cases}\\
   \pi_w(a)=& a =q_am+r_a.
\end{aligned}
\end{equation*}
Note that if $\pi_l(z)$ is the identity permutation, the corresponding edge messages $e_i^{(l)}$ and $e^{(l)}$ are linear functions of the source messages $((x_i)_{i\in [m+1]},z)$.
By the decoding scheme given in the proof of Lemma IV.2 of \cite{connelly2017class}, reassigning the permutation for each $l\in \{1,\dots,w\}$ from the set of permutations $\{\pi_l(a):l=1,\dots,w\}$ would not affect the validity of decoding.
Now, for edges carrying messages indexed by $l$, we can set $\pi_l$ to be the identity permutation (originally $\pi_w(a)$), such that the encoding function on the edges carrying $e_i^{(l)}$ and $e^{(l)}$ turn to be linear and in particular CWL. 

\subsubsection{Subnetwork: $N_3(m_1,m_2)$}
By Lemma V.4 of \cite{connelly2017class}, a solution for $N_3(m_1,m_2)$ ($m_1=m$ and $m_2=s m^{\alpha}$) is given as a single blocklength code over the ring $\mathbf{Z}_{m^{\alpha+1}}$ for each $l=1,2$ by 
\begin{equation*}
\begin{aligned}
   e_0^{(l)} &= \stackrel[j=1]{m_l}{\sum}x_j^{(l)}, \\e_i^{(l)} &= \pi_l(z)+\stackrel[j=1,j\neq i]{m_l}{\sum}x_j^{(l)}\ (i=1,2,\dots,m+1),\\ e^{(l)} &= \pi_l(z)+\stackrel[j=1]{m_l}{\sum}x_j^{(l)}.
\end{aligned}
\end{equation*}
By the encoding functions given above, we know for $l=1,2$ that $e_0^{(l)}$ is a linear function of source messages.
By the proof of Lemma V.2 of \cite{connelly2017class}, permutations $\pi_1$, $\pi_2$ of $\mathbf{Z}_{m^{\alpha+1}}$ are defined as follows. For each $\mathbf{Z}_{m^{\alpha+1}}$, let
\begin{equation*}
    a=\stackrel[{i=0}]{\alpha}{\sum}m^i a_i
\end{equation*}
denote the base $m$ representation of $a$. \cite{connelly2017class} define 

\begin{equation*}
\begin{aligned}
    \pi_1(a)=&m^\alpha a_0 +\stackrel[{i=1}]{\alpha}{\sum}m^{i-1} a_i\\ \pi_2(a)=& a= \stackrel[{i=0}]{\alpha}{\sum}m^i a_i.
\end{aligned}
\end{equation*}

As $\pi_2$ is the identity permutation, the encoding functions on edges carrying $\{e_0^{(2)},e_i^{(2)},e^{(2)}:i=1,2,\dots,m+1\}$ are linear.

Now we show an alternative solution that have linear encoding functions on $\{e_0^{(1)},e_i^{(1)},e^{(1)}:i=1,2,\dots,m+1\}$.
We define the permutations as follows
\begin{align}
     \pi_1(a)=& a= \stackrel[{i=0}]{\alpha}{\sum}m^i a_i\label{equa:pi_1new}\\ \pi_2(a)=& a_\alpha  +\stackrel[{i=0}]{\alpha-1}{\sum}m^{i+1} a_i\label{equa:pi_2new}
\end{align}

By the new definition, $\pi_1(a)$ is the identity permutation, and thus the encoding functions on edges carrying $\{e_0^{(1)},e_i^{(1)},e^{(1)}:i=1,2,\dots,m+1\}$ are linear.
Now we show the validity of the decoding scheme regarding to the new encoding functions. 
That is, by \cite{connelly2017class} for each $a\in\mathbf{Z}_{m^{\alpha+1}}$, it suffices to show that the mapping given by
\begin{equation*}
    a\mapsto (m\pi_1(a), s m^\alpha\pi_2(a))
\end{equation*}
is injective. For each $a,b\in \mathbf{Z}_{m^{\alpha+1}}$, suppose
\begin{align}
     m\pi_1(a)=&m\pi_1(b),\ (\mod\ m^{\alpha+1})\label{equa:pi_1m}\\
     sm^\alpha\pi_2(a)=&sm^\alpha\pi_2(b),\ (\mod\ m^{\alpha+1})\label{equa:pi_2sm}    
\end{align}
where $a=\stackrel[{i=0}]{\alpha}{\sum}m^i a_i$ and $b=\stackrel[{i=0}]{\alpha}{\sum}m^i b_i$. Then from (\ref{equa:pi_1new}),(\ref{equa:pi_1m}) we have
\begin{equation*}
    \stackrel[{i=0}]{\alpha-1}{\sum}m^{i+1} a_i=\stackrel[{i=0}]{\alpha-1}{\sum}m^{i+1} b_i \left(\mod\ m^{\alpha+1} \right)
\end{equation*}
Therefore as $0\leq a_i,b_i<m$
\begin{equation*}
    a_i=b_i\ (i=0,1,\dots,\alpha-1)
\end{equation*}
and from (\ref{equa:pi_2new}) and (\ref{equa:pi_2sm})
\begin{equation*}
    sm^{\alpha}a_\alpha = sm^{\alpha}b_\alpha\ (\mod\ m^{\alpha+1}).
\end{equation*}
Thus, from gcd($m,s$)=1
\begin{equation*}
    m^{\alpha}a_\alpha = m^{\alpha}b_\alpha\ (\mod\ m^{\alpha+1}).
\end{equation*}
We conclude as $0\leq a_\alpha,b_\alpha<m$ that
\begin{equation*}
    a_\alpha = b_\alpha.
\end{equation*}
Thus $a=b$ and we have a valid code with linear encoding functions on $\{e_0^{(1)},e_i^{(1)},e^{(1)}:i=1,2,\dots,m+1\}$.


\section*{VI. The $\epsilon$-error problem on Abelian group network codes}

In this section we follow the definitions and models in \cite{wei2017effect}.

\begin{theorem} \label{theorem:group}
Let $\mathcal{I}$ be a network coding instance.
Let $\epsilon < 1/2$.
If $\mathcal{I}$ is $(\epsilon,{\bf R},n)$-feasible by a group network code then $\mathcal{I}$ is $(0,{\bf R},n)$-feasible by a group network code. 
\end{theorem}


Fix any $0< \epsilon <1/2$. We start with an $(\epsilon,{\bf R},n)$-feasible group network code for network instance $\mathcal{I}$. This group code can be characterized by a finite group $G$ and its subgroups.  We define $G_\mathcal{S}=\cap_{i\in\mathcal{S}} G_i$. Without loss of generality, we assume $|G_\mathcal{S}|=1$ ($G_\mathcal{S}=\{\mathbf{i}\}$) \cite{wei2017effect}, here we denote the identity element of $G$ as $\mathbf{i}$.

We define the binary operation ``$\cdot$'' on the group $G$. For simplicity of notations we neglect the symbol in the following sections. We may represent singletons $\{\mathbf{i}\}$ as $\mathbf{i}$.

Let $t$ be a terminal who demands $i' \in \mathcal{S}$, let $\phi_{In(t)}$ be the global encoding function that encodes the source message vector $X_{\mathcal{S}}=(X_i)_{i\in\mathcal{S}}$ to the incoming message $X_{In(t)}$. We denote the original decoding function as $\phi_t$. We assume the encoding functions are group characterizable and there is no restriction on the decoding functions.

By the definition of group characterization, the source message vector $X_{\mathcal{S}}$ and the incoming message (of terminal $t$) $X_{In(t)}$ can be characterized by subgroups $G_\mathcal{S}=\mathbf{i}$ and $G_{In(t)}$, respectively. Namely, the source message $x_{\mathcal{S}}$ and the incoming message $x_{In(t)}$ corresponds to the coset $g G_\mathcal{S}(=g)$ and $g G_{In(t)}$ respectively.

By our assumption, all sources random variables are uniform. Thus, for an element $g$ uniformly distributed over $G$, for each $i\in\mathcal{S}$, the source random variable is distributed uniformly over the coset of $G_i$ which contains $g$.

The encoding function takes $g G_{\mathcal{S}}=g$ as the input and outputs $g G_{In(t)} \supset g G_{\mathcal{S}}$.

All in all, the communication process to terminal $t$ can be described as follows:
\begin{equation*}
g \xrightarrow{\phi_{In(t)}} g G_{In(t)}\xrightarrow{\phi_t} \bar{g} G_i
\end{equation*}

\begin{claim} \label{claim:Gpre1}
In the context of group network codes, if $G_{In(t)}\not\subset G_i$, then for any decoding scheme the error probability is at least $1/2$.
\end{claim}

\begin{proof}


Let $q=|G_{In(t)}|/|G_{In(t)}\cap G_i|$. By Lagrange's theorem, $q$ must be an integer larger or equal than $2$. 

Thus, each coset of $G_{In(t)}$ contains $q$ different cosets of $G_{In(t)}\cap G_i$ which in turn implies that each coset of $G_{In(t)}$ intersects with $q$ different cosets of $G_i$. 

Now we show that conditioned on an incoming message to therminal $t$, the $q$ messages corresponding to $X_i$ are sent with equal probabilities.
Let $g$ be the message vector sent from the sources, $g_t G_{In(t)}$ be an incoming message received by the terminal $t$ and $g_i G_i$ be a message sent by source $i$.

If $(g_i G_i)\cap (g_t G_{In(t)})=\phi$, we have
\begin{equation*}
    \Pr(g_i G_i\ is\ sent | g_t G_{In(t)}\ is\ received)=0.
\end{equation*}

Since $g$ is uniform in $G$, we have
\begin{equation*}
\begin{aligned}
    &\Pr(g_i G_i\ is\ sent | g_t G_{In(t)}\ is\ received)\\
    =&\Pr (g\in g_i G_i|g\in g_t G_{In(t)})\\
    =&\frac{\Pr(g\in g_i G_i,g\in g_t G_{In(t)})}{\Pr(g\in g_t G_{In(t)})}\\ 
    =&\frac{\Pr(g\in g_i G_i\cap g_t G_{In(t)})}{\Pr(g\in g_t G_{In(t)})}\\ 
    =&\frac{|G_{In(t)}\cap G_i|/|G|}{|G_{In(t)}|/|G|}=1/q
\end{aligned}
\end{equation*}
At terminal $t$, the decoding function needs to determine the output among the $q$ different messages $X_i$, which are conditionally uniform given the incoming message. Thus the optimal correct decoding probability is $1/q$.
As $q\geq 2$, for any decoding scheme, the error probability is at least $1/2$.

\end{proof}

\begin{claim} \label{claim:Gpre2}
In the context of group network codes, if $G_{In(t)}\subseteq G_i$, then there exists a decoding scheme $\phi_t'$ with zero decoding error probability. 
\end{claim}

\begin{proof}
Since $G_{In(t)}\subseteq G_i$, we can design a decoding scheme that takes the coset of $G_{In(t)}$ as input and outputs the coset of $G_i$ which includes $G_{In(t)}$. For such scheme, the decoding error probability is zero.

\end{proof}

With the results above, we prove Theorem~\ref{theorem:group}.

\begin{proof}
Consider a group network code for $\mathcal{I}$ that is $(\epsilon,{\bf R},n)$-feasible for which $\epsilon \in (0,1/2)$. We discuss two scenarios for the given network code.

We first show that it cannot be the case that there exists a terminal decoding $X_i$ for which $G_{In(t)}\not\subset G_i$.
As in that case, by Claim~\ref{claim:Gpre1}, for any decoder, we know that an error occurs with decoding error probability at least $1/2$ which contradicts our assumption that $\epsilon<\frac{1}{2}$.

Otherwise, for all terminals by Claim~\ref{claim:Gpre2} there exists a decoding scheme with zero error probability.

\end{proof}


\end{document}